\documentclass[11pt]{article}
\usepackage[letterpaper,margin=1in]{geometry}
\usepackage{notations}
\definecolor{niceRed}{RGB}{190,38,38}
\definecolor{Red2}{RGB}{219, 50, 54}
\definecolor{mgreen}{HTML}{9ECA8D}
\definecolor{blueGrotto}{HTML}{059DC0}
\definecolor{limeGreen}{HTML}{81B622}
\definecolor{myellow}{rgb}{0.88,0.61,0.14}
\definecolor{navyBlueP}{HTML}{03468F}
\definecolor{Sepia}{HTML}{7F462C}
\definecolor{red2}{HTML}{1F462C}
\definecolor{orange2}{HTML}{FF8000}
\definecolor{mgray}{HTML}{ABB3B8}
\definecolor{myPurple}{RGB}{175,0,124}
\definecolor{royalBlue}{HTML}{057DCD}
\definecolor{mpink}{HTML}{FC6C85}

\hypersetup{
	colorlinks = true,
	linkcolor = niceRed,
	citecolor = royalBlue,
	linktocpage = true,
	urlcolor = mpink
}

\title{Feature-Based Online Bilateral Trade}

\author{
    Solenne Gaucher$^\#$\quad
    Martino Bernasconi$^\dagger$ \quad
    Matteo Castiglioni$^\ddagger$ \\
    Andrea Celli$^\dagger$ \quad
    Vianney Perchet$^{\#,\ast}$  \vspace{6mm}\\
    $^\dagger$\ Bocconi university\\
    $^\ddagger$\ Politecnico di Milano\\
    $^\ast$\ Criteo AI Lab\\
    $^\#$\ ENSAE\vspace{2mm}\\
    {\textcolor{black}{\small\texttt{\{martino.bernasconi,andrea.celli2\}@unibocconi.it}, \quad \texttt{matteo.castiglioni@polimi.it,}}}\\
    {\textcolor{black}{\small\texttt{vianney.perchet@normalesup.org}\quad\small\texttt{solenne.gaucher@ensae.fr}}}
}
\date{}
\begin{document}

\maketitle
\thispagestyle{empty}

\begin{abstract}
    Bilateral trade models the problem of facilitating trades between a seller and a buyer having private valuations for the item being sold.
    In the online version of the problem, the learner faces a new seller and buyer at each time step, and has to post a price for each of the two parties without any knowledge of their valuations. 
    We consider a scenario where, at each time step, before posting prices the learner observes a context vector containing information about the features of the item for sale. The valuations of both the seller and the buyer follow an unknown linear function of the context. In this setting, the learner could leverage previous transactions in an attempt to estimate private valuations. 
    We characterize the regret regimes of different settings, taking as a baseline the best context-dependent prices in hindsight. First, in the setting in which the learner has two-bit feedback and strong budget balance constraints, we propose an algorithm with $O(\log T)$ regret. Then, we study the same set-up with noisy valuations, providing a tight $\widetilde O(T^{\nicefrac23})$ regret upper bound.
    Finally, we show that loosening budget balance constraints allows the learner to operate under more restrictive feedback. Specifically, we show how to address the one-bit, global budget balance setting through a reduction from the two-bit, strong budget balance setup.
    This established a fundamental trade-off between the quality of the feedback and the strictness of the budget constraints. 
\end{abstract}

\clearpage

\section{Introduction}

Bilateral trade models scenarios in which a seller and a buyer, both having a private valuation for a good, are interested in trading it in an attempt to maximize their respective utilities \cite{Vickrey61,MyersonS83}.
%
%
We study the online bilateral trade problem introduced by  \citet{CesaBianchiCCF21}. At each time $t$, a new seller and buyer arrive, each with private valuations $s_t$ and $b_t$, respectively. The seller's valuation $s_t$ is the lowest price they are willing to accept for the item. Analogously, the buyer's valuation $b_t$ represents the highest price they are willing to pay for the item. 
The learner, without any knowledge about the private valuations at the current time $t$, posts two (possibly randomized) prices: $p_t$ to the seller and $q_t$ to the buyer. A trade happens when $s_t \le p_t$ and $ q_t \le b_t$, so both agents agree to trade. 
The \textit{gain from trade} for a pair of prices $(p,q)$ at time $t$ is 
\begin{equation}\label{eq:gft def}
    \gft_t(p,q) \defeq \ind{s_t \le p} \ind{q\le b_t} (b_t - s_t).
\end{equation}
This corresponds to the increase in social welfare generated by the trade. The goal of the learner is to maximize the overall gain from trade or, equivalently, minimize the regret with respect to the best policy in hindsight.
A key challenge in the online bilateral trade problem is the inherently limited nature of the feedback: under \textit{two-bit feedback}, the learner receives feedback $(\indicator{s_t\le p_t},\indicator{b_t\ge q_t})$, and under \textit{one-bit feedback} the learner only observes whether the trade happened or not, that is $\indicator{s_t\le p_t}\cdot \indicator{b_t\ge q_t}$. Both types of feedback are significantly less informative than the traditional bandit feedback, since the learner cannot even reconstruct the GFT received for the prices it posted.

In the standard model of bilateral trade, the platform lacks information about the seller, buyer, or the item being sold. However, this scenario is unrealistic in practice, where some information is usually available to the learner. In contexts like online marketplaces, products are often highly differentiated, and the learner can observe some product features upfront. Then, the learner could base pricing decisions on the features of the current product. This would enable the learner to utilize past price and trade data to estimate the values associated with each feature, informing future pricing decisions.

We introduce the \textit{feature-based} online bilateral trade problem, in which the learner observes a feature vector $x_t\in\R^d$ before posting prices for round $t$. Following the classic feature-based dynamic pricing set-up (see, \eg \cite{Cohen_feature_DP,javanmard2017perishability,javanmard2019dynamic,keskin2014dynamic,xu2021logarithmic}), we study the setting in which private valuations are of the form $x_t^\top \theta + \xi_t$, where $\theta\in\R^d$ is an unknown vector denoting the importance of each feature, and $\xi_t$ is an i.i.d.~noise term. Feature vectors $x_t$ are chosen adversarially. This ensures that our solution is robust to scenarios where features are correlated and where the set of relevant features evolves over time.

\subsection{Our contributions}
In this paper, we introduce the feature-based online bilateral trade model, and characterize the regret for various scenarios with adversarially generated feature vectors. We provide the following results for the case in which the learner has two-bit feedback and strong budget balance constraints (the learner has to set the same price to the seller and to the buyer, \ie $p_t = q_t$ at each round $t$): 
\begin{OneLiners}
\item First, we consider the deterministic setup where, at each $t$, the seller's valuation is $s_t=x_t^\top \theta^s$ and the buyer's valuation is $b_t=x_t^\top\theta^b$. We show that in this case, it is possible to adapt techniques proposed by \citet{Cohen_feature_DP} in the context of feature-based dynamic pricing. The main difference here is that we need to maintain separate ellipsoidal uncertainty sets for the seller and the buyer, respectively. Our analysis yields a regret of order $O(\log T)$ (\Cref{thm:no_noise_2_bits}).

\item Second, we consider the case of noisy valuations in which $s_t=x_t^\top\theta^s +\xi_t^s$ and $b_t=x_t^\top\theta^b+\xi_t^b$, where $\xi_t^s$ and $\xi_t^b$ are i.i.d.~noise terms independent from $x_t$, with bounded support and densities. 
We start by decomposing the expected gain-from-trade into components that can be individually estimated using two-bit feedback (\Cref{lem:EGFT}). Then, we describe an explore-\emph{or}-commit (EOC) algorithm with regret $\widetilde O(T^{\nicefrac34})$ (\Cref{thm:etc}) (this algorithm will be employed in the subsequent reduction to the one-bit setting). 
Finally, we devise a \textsc{Scouting Bandits with Information Pooling} algorithm, which makes a more efficient use of the information collected during the learning phase. This algorithms achieves a regret $\widetilde O(T^{\nicefrac23})$(\Cref{thm:noise_2_bits}),
which is minimax optimal up to poly-logarithmic factors and dependence on $d$. Indeed, there exists a matching $\Omega(T^{\nicefrac23})$ lower bound for the stochastic bilateral trade problem without features in the case of independent valuations with bounded densities and support \cite{CesaBianchiCCF21}.
\end{OneLiners}

Finally, we provide a general reduction that demonstrates how the difficulty of maintaining a per-round budget balance can be traded for the ability to operate under more demanding feedback conditions.
In particular, given a two-bit explore-or-commit algorithm guaranteeing strong budget balance and sublinear regret, we show that it is possible to construct a no-regret algorithm that works under one-bit feedback, and is global budget balanced (\Cref{thm: one bit regret}). The one-bit regret guarantees are dependent on a natural measure of the social welfare generated by the market.

\subsection{Related works}

\xhdr{Bilateral trade.} In the offline setting, \citet{MyersonS83} showed the existence of instances where a fully efficient mechanism that satisfies incentive compatibility, individual rationality, and budget balance does not exist. Subsequent research focused on finding approximately efficient mechanisms in the Bayesian setting \cite{BlumrosenD14,kang22fixed,Mcafee08,BlumrosenM16,brustle2017approximating,DengMSW21,fei2022improved}.

\xhdr{Online bilateral trade.} Cesa-Bianchi et al. \cite{CesaBianchiCCF21,cesa2024bilateral} study the case in which valuations are drawn i.i.d.~from some fixed unknown distribution and the learner has to enforce strong budget balance. They provide sublinear regret guarantees in the full-feedback setting, and under partial feedback when valuations are i.i.d.~samples from a smooth distribution, independently for the seller and the buyer. 
If the learner is only required to enforce \textit{weak budget balance} (\ie $p_t \le q_t$ for each $t$), then \citet{AzarFF22} provide an algorithm achieving a tight sublinear $2$-regret when the sequence of valuation is generated by an oblivious adversary.
\citet{CesaBianchiCCFL23} show that sublinear regret can be achieved beyond the i.i.d.~stochastic setting under a $\sigma$-smooth adversary model. \citet{bernasconi2023no} show that sublinear regret can be achieved in the fully adversarial setting if the learner enforces global budget balance constraints (\ie the constraint has to hold over the entire time horizon).

\xhdr{Feature-based pricing.} The one-dimensional version of the problem was introduced by \citet{Kleinberg2003}, while \citet{amin2014repeated} introduced the problem in the contextual, multi-dimensional set-up under i.i.d.~contexts. \citet{Cohen_feature_DP} study a model with adversarial contexts for which they provide $\widetilde O(d^2\log T)$ regret guarantees in the noiseless setting, improving over the $\widetilde O(\sqrt{T})$ guarantees obtainable with general-purpose contextual bandits algorithms \cite{agarwal2014taming}. Moreover, they provide $\widetilde{O}(T^{\nicefrac{2}{3}})$ regret guarantees for scenarios with adversarial contexts and additive noise generated from a known sub-Gaussian distribution.
Further improvements in the achievable rates were later provided by \citet{lobel2018multidimensional,leme2018contextual,liu2021optimal}. Further instantiations of the feature-based pricing model can be found in the survey of \citet{den2015dynamic}.

\section{Preliminaries}

\xhdr{Notations.}
For any $k\ge 1$, we compactly denote the set $\{1,2,\ldots,k\}$ as $\range{k}$.
We denote by $\LJ(H)$ the L\"owner-John ellipsoid of a convex body $H$ (\ie the minimal volume ellipsoid that contains $H$). 
For a positive definite matrix $\bM\in\R^{d\times d}$ and $x\in\R^d$, we denote $\sqrt{x^\top\bM x}$ by $\|x\|_{\bM}$.
In our proof sketches, we write $A_T\lesssim B_T$ if there exists a problem-dependent constant $c$ such that $A_T \leq c B_T$.

\xhdr{Learning protocol.}
At each round $t$, a new buyer and a seller arrive, characterized by private valuations $b_t$ and $s_t$, respectively. After observing a feature vector $x_t\in\R^d$, the learner posts two prices: $p_t$ to the seller and $q_t$ to the buyer. The trade happens if both agents accept the proposed prices, that is $s_t\le p_t$ and $q_t\le b_t$. When the trade happens the learner is awarded the gain from trade $\gft_t(p_t,q_t)$, defined as per \Cref{eq:gft def}. Feature vectors $x_t$ are generated by an oblivious adversary.

\xhdr{Valuations.} We consider two models to describe how features impact on the seller's and buyer's valuations. In the first setting, valuations are a \textbf{deterministic} function of the context. In particular, at each round $t\in\range{T}$, the valuations of the seller and of the buyer are given by 
\begin{equation}\label{eq:2bit_no_noise}
s_t = x_t^{\top}\theta^s\quad\textnormal{ and }\quad b_t=x_t^\top\theta^b,
\end{equation}
respectively, where $\theta^s,\theta^b\in\R^d$ are unknown to the learner. We assume that both the contexts $x_t$ and the parameters $\theta^b$ and $\theta^s$ are bounded.
\begin{assumption}\label{ass:bounded_context}
There exist $A,B\ge 0$ such that $\max(\Vert \theta^s\Vert, \Vert \theta^b\Vert) \leq A$, and $\Vert x_t\Vert \leq B$.
\end{assumption}
\vspace{-0.5mm}
In the second setting, we consider \textbf{noisy} valuations. In particular, we assume that at each round $t$ the valuations are given by
\begin{equation}
    b_t = x_t^{\top}\theta^b + \xi_t^b \quad \text{ and }\quad s_t = x_t^{\top}\theta^s + \xi_t^s,\label{eq:2bit_noise}
\end{equation}
where $\xi_t^b$ and $\xi_t^s$ are i.i.d.~centered noise terms independent from the context $x_t$, with densities $f^b$ and $f^s$, respectively. We denote by $F^s$ the c.d.f.~of $\xi_t^s$, and by $D^b$ the demand function of $\xi_t^b$: for all $x\in \mathbb{R}$, $D^b(x) \defeq \mathbb{P}\left(\xi_t^b \geq x\right)$. We make the following assumption on the densities of $\xi^s_t$ and $\xi^b_t$.
\begin{assumption}\label{ass:bounded_densities}
    Both $\xi_t^s$ and $\xi_t^b$ have bounded support in $[-C,C]$, and have densities bounded by $L$.
\end{assumption}
\vspace{-0.5mm}
We observe that under Assumption \ref{ass:bounded_context} and \ref{ass:bounded_densities}, the valuations of the seller and the buyer are bounded in $[-P, P]$, where $P \defeq C + AB$.\footnote{
For simplicity, we work within $[-P, P]$. If non-negative prices are required, a translation is sufficient.} The bounded densities assumption is standard in repeated bilateral trade problems with stochastic valuations, and there exists linear regret lower bounds for the case in which this assumption is lifted (see \citet[Theorem 6]{cesa2024bilateral}).

\xhdr{Feedback models.} After posting prices $(p_t,q_t)$, the learner does not observe directly the $\gft$ or the valuations, but receives some feedback $z_t$ about the transaction. We focus on two feedback models: (i) \textit{Two-bit feedback}, where both the buyer and the seller reveal their willingness to accept the prices offered by the learner (\ie $z_t$ is composed by the two bits $(\ind{s_t \le p_t}, \ind{q_t \le b_t})$); (ii) \textit{One-bit feedback}: the learner only observes whether the trade happened or not (\ie $z_t = \ind{s_t \le p_t} \cdot\ind{q_t \le b_t}$). Both feedback models have the desirable property of revealing minimal information about agent's private valuations. 


\xhdr{Regret.} Our objective is to develop dynamic policies that perform well in terms of minimizing regret. In the deterministic setting, the worst-case regret for learning algorithm $\algA$ is defined as 
\[
R_T(\algA)\defeq\max_{\substack{(\theta^s,\theta^b)\in [-A,A], \vx\in[-C,C]^T}} \sum_{t\in\range{T}}\mleft(\mleft[x_t^\top(\theta^b-\theta^s)\mright]^+-\gft_t(p_t,q_t)\mright).
\]
In the noisy model, we compare against a benchmark that maximizes the expected gain from trade. Let $\egft(x, p, q)\defeq \mathbb{E}\left[\gft_t(p,q)\,\vert\, x_t = x\right]$. The pseudo-regret for the noisy setting is defined as
\[
R_T(\algA) \defeq \sum_{t\in\range{T}}\max_{\substack{(p,q)\in[-P,P]^2\\p\le q}} \egft(x_t, p,q) - \sum_{t\in\range{T}}\egft(x_t, p_t, q_t).
\]
A property that directly follows by definition is that, for any feature vector, there exists a pair of identical prices that maximize the expected gain from trade. To simplify the notation, we omit the argument $q$ of $\gft_t$ and $\egft$ if the same price is posted to both agents.
\section{Warm-up: two-bit feedback, strong budget balance, no noise}\label{sec:deterministic}

In this section, we consider the simplest set-up in terms of the information available to the learner, by focusing on the setting with two-bit feedback and deterministic valuations following \Cref{eq:2bit_no_noise}. We show that it is possible to adapt the \textsc{EllipsoidPricing} algorithm, originally proposed by \citet{Cohen_feature_DP} for feature-based dynamic pricing, to our bilateral trade problem. The key distinction lies in managing separate uncertainty sets for the seller and the buyer, while carefully setting prices consistent with both estimations. In this set-up, we can update these uncertainty sets separately.



\Cref{alg:EPBT} describes the main step of \textsc{EllipsoidPricing for Bilateral Trade (EP-BT)}.  
The algorithm maintains two ellipsoidal uncertainty sets $E^s_t$ and $E^b_t$ (see \Cref{app:ellipsoid} for some basic facts about ellipsoids). The key idea is that, at each update, the volume of an uncertainty shrinks ``fast enough'' to yield a good estimate of the true parameter in a small number of iterations.
The problems of computing the maximum and minimum possible valuations of the seller and buyer (Line \ref{algline:highest lowerst val}) admit a simple closed-form solution (see, \eg, \citet[Chapter 3]{grotschel2012geometric}).
%
%
If the smallest possible valuation for the buyer $\underline{b_t}$ is above the highest possible valuation for the seller $\overline{s_t}$ (Line \ref{algline:exploit}), the algorithm can post any price between these two values, ensuring that the trade will occur. 
If $\overline{s_t}$ and $\underline{s_t}$ are far apart (Line \ref{algline:explore_s}), the algorithm performs a binary search step (\textit{a.k.a.} ``explore'' step), and posts a price $p_t$ which is halfway between the seller's minimum and the maximum possible valuations.
Then, we update the seller's uncertainty set as follows: if the seller accepts the sale, it implies that $x_t^{\top}\theta^s \leq p_t$, and we define the half-ellipsoid $K_{t+1}^s = E^s_t \cap \{\theta :  x_t^{\top}\theta^s \leq p_t\}$. Otherwise, we set $K_{t+1}^s = E^s_t \cap \{\theta :  x_t^{\top}\theta^s > p_t\}$. Finally, we round this half-ellipsoid by replacing it by its L\"owner-John ellipsoid, \ie the ellipsoid containing $K_{t+1}^s$ with the smallest volume. We proceed similarly for the buyer (Line \ref{algline:explore_b}).
If the highest and lowest possible valuations of both the buyer and seller are close to each other, and the buyer's minimum valuation is lower than the seller's maximum valuation, then the difference of valuations between the two parties is small, resulting in a negligible gain from trade. Therefore, in this case, we can safely post an arbitrary price (Line \ref{algline:random_price}).

\begin{algorithm}[b]
\caption{\textsc{EllipsoidPricing for Bilateral Trade (EP-BT)}}\label{alg:EPBT}
\begin{algorithmic}[1]
\State \textbf{Input} parameter $\epsilon>0$, bound $A$.
\State \textbf{Initialize} $K^{s}_1, K^{b}_1 \gets$ $d$-dimensional ball of radius $A$, $E^s_1 = \LJ(K^{s}_1)$ and $E^{b}_1=\LJ(K^{b}_1)$.
\For{ $t\in\range{T}$}
\State Set $\textstyle\underline{s_t} = \min_{\theta \in E^s_t}x_t^{\top}\theta\,\,\,\textnormal{and}\,\,\,\overline{s_t} = \max_{\theta \in E^s_t}x_t^{\top}\theta$ and compute $(\underline{b_t},\overline{b_t})$ analogously\label{algline:highest lowerst val}
\If{$\overline{s_t} < \underline{b_t}$} post price $p_t = (\overline{s_t}+\underline{b_t})/2$\label{algline:exploit}
\ElsIf{$\overline{s_t} - \underline{s_t} \geq \epsilon$}\label{algline:explore_s}
\State Post price $p_t = (\overline{s_t}+\underline{s_t})/2$\label{algline:explore_s_beg}
\State Update $K^s_{t+1}$ according to the seller's feedback, set $E^s_{t+1} = \LJ(K^s_{t+1})$ \label{algline:explore_s_end}
\ElsIf 
{$\overline{b_t} - \underline{b_t} \geq \epsilon$}\label{algline:explore_b}
\State Post price $p_t = (\overline{b_t}+\underline{b_t})/2$
\State Update $K^b_{t+1}$ according to the seller's feedback, set $E^b_{t+1} = \LJ(K^b_{t+1})$
\Else \label{algline:random_price} post price $p_t = (\overline{s_t}+\underline{b_t})/2$
\EndIf
\EndFor
\end{algorithmic}
\end{algorithm}
\begin{restatable}{theorem}{noNoiseEllipsoid}\label{thm:no_noise_2_bits}
Assume that the seller's and buyer's valuations follow \Cref{eq:2bit_no_noise}. Under \Cref{ass:bounded_context}, \Cref{alg:EPBT} with $\epsilon = ABd^2/T$ has regret bounded by
$
    R_T \leq 10ABd^2\log\left(20(d+1)Td^{-2}\right).
$
\end{restatable}
We observe that the algorithm explores only whenever the estimations are not precise enough. In the proof (\Cref{app:2bitnonoise}), we show that this happens in at most $O(\log T)$ rounds. 
\section{Noisy valuations with two-bit feedback}
\label{sec:noisy}

In this section, we consider bilateral trade with two-bit feedback and noisy valuations. We introduce key components that we exploit in the analysis, before integrating them into an explore-\emph{or}-commit framework which that ensures a regret $\widetilde{O}(T^{\nicefrac34})$. Finally, we show a careful re-design of the last phase of the algorithm yields a tight regret upper bound of $\widetilde{O}(T^{\nicefrac23})$. 

We rely on the following decomposition lemma, which extends Lemma 1 of \citet{cesa2024bilateral} to contextual settings. The proof can be found in \Cref{aux:egft dec}.

\begin{restatable}{lemma}{egftdec}\label{lem:EGFT} Under Assumptions \ref{ass:bounded_context} and \ref{ass:bounded_densities}, the expected gain from trade for $x$ at price $p$ is given by 
\begin{align}\label{eq:lemma_EGFT}
    \egft(x, p) = I(\delta^b) F^s(\delta^s) + J(\delta^s)D^b(\delta^b),
\end{align}
where $\delta^s \defeq p-x_t^{\top}\theta^s$, $\delta^b \defeq p - x_t^{\top}\theta^b$, $I(\delta) \defeq\int_{\delta}^C D^b(u)du$, and $J(\delta) \defeq \int_{-C}^{\delta}F^s(u)du$.
\end{restatable}

\Cref{lem:EGFT} highlights the principal difficulties of our setting. 
It emphasizes that the expected gain from trade for a price $p$ depends on the pair of price increments $(\delta^s, \delta^b)$, or equivalently, on the pair $(\delta^s, \Delta_t)$, where the difference in average valuations $\Delta_t$ is defined as $\Delta_t \defeq x_t^{\top}(\theta^b - \theta^s)$.
This shows that items with the same $\Delta_t$ have the same optimal seller's price increment $\delta^\ast_t$: if $p = x^{\top}_t\theta^s + \delta^*_t$ maximizes $\egft(x_t,\cdot)$, and if $x_{t'}^{\top}(\theta^b - \theta^s) = x_t^{\top}(\theta^b - \theta^s)$, then $p' = x_{t'}^{\top}\theta^s + \delta^*_t$ maximizes $\egft(x_{t'},\cdot)$.
However, knowing the optimal increment $\delta^\ast_t$ for a specific $\Delta_t$ is not sufficient to determine the optimal increment for a different $\Delta_t$ (see \Cref{sec:example}). Thus, finding the optimal price increment across various $\Delta_t$ values might require us to precise estimate the reward function over a broad range of increments.
This is a notable departure from the non-contextual stochastic bilateral trade problem discussed in \citet{cesa2024bilateral}, where precise estimation of the reward function around the optimal price suffices.
Since accurately estimating functions such as $F^s$, $D^b$, $I$, and $J$ across a wide range of arguments complicates the problem, the regret might be higher than the $\widetilde{O}(T^{\nicefrac23})$ rate achieved in \citet{cesa2024bilateral}. Surprisingly, we show that with careful coordination of the various estimation procedures, the optimal rate for the non-contextual case can still be attained.

\subsection{Building blocks: subroutines for learning parameters}

We begin by presenting the sub-routines to estimate the parameters $\theta^s$ and $\theta^b$, as well as the functions $F^s$, $D^b$, $I$ and $J$ used in the $\egft$ decomposition of  \cref{lem:EGFT}. These sub-routines serve as the base components for the subsequent algorithms.
\begin{algorithm}
\caption{\textsc{Estimation Subroutines}}\label{alg:est par}
\begin{algorithmic}[1]
\Procedure{\textsc{Est-Par}}{$\cTpar$}\Comment{Update estimate of parameters $\theta^b,\theta^s$}
\State Draw price $p_t \sim \cU([-P,P])$ and post $(p_t, p_t)$
\State $\bV \gets \sum_{l\in \cTpar} x_l x_l^{\top}+  \bI_d$ 
\State $\widehat{\theta}^s \gets 2P\bV^{-1}\sum_{l\in \cTpar}\left({\indicator{p_l\le s_l}} - \frac{1}{2}\right)x_l$\label{algline:update theta s}
\State $\widehat{\theta}^b \gets 2P\bV^{-1}{\sum}_{l\in \cTpar} \left({\indicator{p_l \leq b_l}}-\frac{1}{2}\right)x_l$
\EndProcedure
\hrulealg
\Procedure{\textsc{Est-Int}}{$\cTint$}\Comment{Update estimate of integrals $I$, $J$}
\State Draw price $p_t \sim \cU([-P,P])$ and post $(p_t, p_t)$
\For{$k \in \cK$}
\State $\widehat{J}^k \gets \frac{2P}{\vert \cTint\vert}\sum_{l \in \cTint}{\mathbb{I}\left\{s_l \leq p_l \leq k\epsilon + x_l^{\top}\widehat{\theta}^s \right\}}.$
\State $\widehat{I}^k \gets \frac{2P}{\vert \cTint\vert}\sum_{l \in \cTint}{\mathbb{I}\left\{k\epsilon + x_l^{\top}\widehat{\theta}^b \leq p_l \leq b_l  \right\}}$
\EndFor
\EndProcedure
\hrulealg
\Procedure{\textsc{Est-F}}{$\cTF_k$,  grid point $k\in \cK$}\Comment{Update estimate of $F^s$}
\State Post price $p_t = x_t^{\top}\widehat{\theta}^s + k\epsilon$. Then update $\widehat{F}^k \gets \sum_{l\in \cTF_k} {\mathbb{I}\left\{s_l \leq p_l\right\}} / \vert \cTF_k \vert$
\EndProcedure
\hrulealg
\Procedure{\textsc{Est-D}}{$\cTD_k$,  grid point $k\in \cK$}\Comment{Update estimate of $D^b$}
\State Post price $p_t = x_t^{\top}\widehat{\theta}^b + k\epsilon$. Then update $\widehat{D}^k \gets \sum_{l\in \cTD_k} {\mathbb{I}\left\{b_l \geq p_l\right\}}/\vert \cTD_k \vert$
\EndProcedure
\end{algorithmic}
\end{algorithm}
The first sub-routine, \textsc{Est-Par}, estimates the parameters $\theta^s$ and $\theta^b$. It relies on the fact that, if $p_t \sim \cU([-P,P])$, then $2P\left(\mathbb{I}\{p_t\geq s_t\} - 1/2\right)$ is an unbiased estimate of $x_t^{\top}\theta^s$ (the case of $\theta^b$ is analogous). Then, we can rely on classical results to estimate the parameters (see \citet[Chapter 19]{lattimore2020bandit}). 
The second sub-routine, \textsc{Est-Int}, estimates the integrals $I(\delta)$ and $J(\delta)$ over a grid of price increments $\{k\epsilon : k \in \cK\}$. The third and fourth sub-routines, \textsc{Est-F} and \textsc{Est-D}, provide estimates of the c.d.f.~$F^s$ and the demand function $D^b$ at $k\epsilon$ for a given increment level $k\in\cK$, respectively.

\subsection{Explore-Or-Commit framework}

First, we consider a natural Explore-\emph{Or}-Commit (EOC) algorithm which uses each round \emph{either} to compute estimates for the terms in \Cref{eq:lemma_EGFT}, \emph{or} to greedily play the empirical best action. This strategy is described in \Cref{alg:explore-commit}.
Our algorithm is not a standard Explore-Then-Commit algorithm. Indeed, it resorts to the estimation subroutine $\textsc{Est-Par}$ (Line \ref{line:estpar etc}) whenever the estimation is not ``good enough'' and not only in the first rounds. 
The other estimation subroutines are executed until a certain number of updates is reached.
We show that the number of necessary exploration rounds is upper bounded by $\widetilde O(T^{\nicefrac34})$. This will be useful in section \Cref{sec:one bit} to handle one-bit feedback.
Finally, in Lines \ref{line:etc exploit  1} and \ref{line:etc exploit  2} the algorithm selects the best price increments according to the current estimates (\ie ``commit'' rounds), and posts price $p_t$ build accordingly to both agents.

\begin{algorithm}
\caption{\textsc{Explore-Or-Commit for Bilateral Trade (EOC-BT)}}\label{alg:explore-commit}
\begin{algorithmic}[1]
\State \textbf{Input}: parameter $\mu$, length of estimation phases $\Tint$, $\TFD$, discretization error $\epsilon$, confidence $\delta$.
\State \textbf{Initialize}: $K = \lceil 2P/\epsilon\rceil + 3$, $ \cK\gets \llbracket -K,K\rrbracket$,  $\cTpar\hspace{-.1cm} =\hspace{-.1cm} \cTint \hspace{-.1cm}=\hspace{-.1cm}\cTF_k\hspace{-.1cm} =\hspace{-.1cm} \cTD_k\hspace{-.1cm}=\hspace{-.1cm} \varnothing$ for all $k\in\cK$, $\bV = \mathbf{I}_{d}$.
\While{$t \leq T$}
\If{$\Vert x_t\Vert_{\bV^{-1}} > \mu$}
$\cTpar\gets \cTpar\cup\{t\}$, $\textsc{Est-Par}(\cTpar)$\label{line:estpar etc}
\ElsIf{$\vert \cTint \vert < \Tint$} 
$\cTint\gets \cTint\cup\{t\}$, $\textsc{Est-Int}(\cTint)$
\ElsIf{ for some $k \in \cK$, $\vert \cTF_k \vert < \TFD$ or $\vert \cTD_k \vert < \TFD$}
\If{$\vert \cTF_k \vert < \TFD$} $\cTF_k\gets \cTF_k\cup\{t\}$, $\textsc{Est-F}(\cTF_k, k)$
\ElsIf{$\vert \cTD_k \vert < \TFD$} $\cTD_k\gets \cTD_k\cup\{t\}$, $\textsc{Est-D}(\cTD_k, k)$
\EndIf
\Else \Comment{Exploitation phase}
\State $\cA_t \gets \mleft\{(k, k') \in \cK^2 : k' = \left\lfloor \mleft(k\epsilon 
- x_t^{\top}\mleft(\widehat{\theta}_t^b - \widehat{\theta}_t^s\mright)\mright)\epsilon^{-1}\right\rfloor\mright\}$\label{line:etc exploit 1}
\State Choose $(k_t, k'_t) \in \argmax_{(k, k') \in \cA_t}  \widehat{F}^k\widehat{I}^{k'}  + \widehat{D}^{k'}\widehat{J}^k$ and post price $p_t = x_t^{\top}\widehat{\theta}^s + k_t\epsilon$\label{line:etc exploit  2}
\EndIf
\EndWhile
\end{algorithmic}
\end{algorithm}



The following theorem bounds the regret of \Cref{alg:explore-commit}. Its proof can be found in Appendix \ref{app:explore_commit}.
\begin{restatable}{theorem}{etc}\label{thm:etc}
Set $\epsilon = (\log(T)/T)^{1/4}$,  $\delta = \left(T(74+32P\epsilon^{-1})\right)^{-1}$, $\mu = \epsilon\big(P\sqrt{d\log\left(\frac{1 + B^2T}{\delta}\right)} + A\big)^{-1}$, $\Tint = \lceil 8P^2\log(1/\delta)\epsilon^{-2}\rceil$, and $\TFD = \lceil 2\log(1/\delta)\epsilon^{-2}\rceil$. Then, under Assumptions \ref{ass:bounded_context} and \ref{ass:bounded_densities}, there exist a constant $\widetilde T$ such that, if $T\geq \widetilde T$,
$
    R_T \leq  \widetilde O( T^{\nicefrac{3}{4}})
$
with probability at least $1-T^{-1}$.
\end{restatable}

\begin{sproof}
    For some $\epsilon >0$ to be chosen later, let $K \approx P/\epsilon$. Using standard arguments \cite{NIPS2011_e1d5be1c,carpentier2020elliptical}, we can show that $\widetilde{O}(\epsilon^{-2})$ samples are sufficient for \textsc{Est-Par} to ensure that, with high probability, for all $t\notin \cTpar$ the errors $\vert x_t^{\top}(\theta^s - \widehat{\theta}^s_t)\vert$ and $\vert x_t^{\top}(\theta^b - \widehat{\theta}^b_t)\vert$ are smaller than $\epsilon$. When this happens, classical concentration arguments ensure that $\widetilde{O}(\epsilon^{-2})$ samples are sufficient for \textsc{Est-Int} to estimate $I(k\epsilon)$ and $J(k\epsilon)$ with precision $\epsilon$ uniformly for $k \in \cK$. By contrast, $\widetilde{O}(\vert \cK\vert \epsilon^{-2})$ samples are necessary for \textsc{Est-F} and \textsc{Est-F} to estimate $F^s(k\epsilon)$ and $D^b(k\epsilon)$ with precision $\epsilon$ uniformly for $k \in \cK$.
    Therefore, with high probability, after approximately $\widetilde{O}(|\mathcal{K}| \epsilon^{-2})$ rounds of exploration, the expected gain from trade for prices $x_t^{\top}\widehat{\theta}^s_t + k\epsilon$ is estimated with precision $O(\epsilon)$ uniformly for $k\in \mathcal{K}$. 
    Furthermore, Assumption \ref{ass:bounded_densities} ensures that the reward function is Lipschitz-continuous, thus the discretization error is also of order $\widetilde{O}(\epsilon)$. By choosing $\epsilon = T^{\nicefrac{-1}{4}}$, we balance the estimation and discretization errors with the regret of the exploration phase, thereby obtaining a regret $\widetilde{O}(T^{\nicefrac34})$.
\end{sproof}

The regret of Algorithm \ref{alg:explore-commit} is primarily driven by the estimation of the c.d.f. $F^s$ and demand function $D^b$ uniformly over a grid of increments. 
We now present an alternative sub-optimal approach, which allows us to introduce key concepts useful in the following section.
Specifically, to leverage the fact that the optimal increment $\delta^s$  only depends on the difference in average valuations $\Delta_t$, we could first execute the sub-routines $\textsc{Est-Par}$ and $\textsc{Est-Int}$, yielding estimates of $x_t^{\top}\theta^s$, $x_t^{\top}\theta^b$, $I$ and $J$ up to precision $\epsilon$ using $\widetilde{O}(\epsilon^{-2})$ samples. Then, we could round the value of $\Delta_t$ on a grid of size $\epsilon^{-1}$, and run independent Scouting Bandit algorithms  (as described in \cite{cesa2024bilateral}) for each of the $\epsilon^{-1}$ rounded values.
%
The grid size implies a discretization error of order $O(\epsilon)$. The highest regret occurs when each of the $\epsilon^{-1}$ independent Scouting Bandit algorithms runs for the same number of rounds $T_{\epsilon} = T\epsilon$. Each algorithm incurs a regret $\widetilde{O}(T_{\epsilon}^{\nicefrac23})$, so their combined regret is 
$\widetilde{O}(T_{\epsilon}^{\nicefrac23} \cdot \epsilon^{-1}) = \widetilde{O}(T^{\nicefrac23}\epsilon^{\nicefrac{-1}{3}})$. By selecting $\epsilon = T^{\nicefrac{-1}{4}}$, this strategy also results in a regret of order $\widetilde{O}(T^{\nicefrac34})$.

\subsection{Closing the gap: scouting bandits with information pooling}

Both the EOC strategy---where the reward function for all potential values of $(\delta^s, \Delta_t)$ is learned before committing to posting greedy prices---and strategies that use independent scouting bandit algorithms for each (rounded) value of $\Delta_t$ result in the same regret. 
Surprisingly, we show that by combining the strengths of both strategies, it is possible to design an algorithm that achieves the optimal regret rate of $\widetilde{O}(T^{\nicefrac23})$. To do this, we design a scouting bandit strategy with \textit{information pooling} across different values of $\Delta_t$. This idea is related to \textit{cross-learning}, developed for bandits with graph feedback \cite{balseiro2019contextual}. In particular, after having estimated $x_t^{\top}\theta^s$, $x_t^{\top}\theta^b$, $I$ and $J$, we run a successive elimination algorithm for each value of $\Delta_t$. Then, to estimate the reward of a price corresponding to increments $\delta^s$ and $\delta^b$, we use the feedback from rounds where these increments have been selected \textit{across all values of $\Delta_t$}. This strategy, called \textsc{Scouting Bandit with Information Pooling (SBIP)}, is described in Algorithm \ref{alg:SBLBT} (a detailed description of the algorithm and proofs are deferred to \Cref{app:2bit noise scouting}).

\begin{algorithm}
\caption{\textsc{Scouting Bandit with Information Pooling (SBIP)}}\label{alg:SBLBT}
\begin{algorithmic}[1]
\State \textbf{Input}: parameter $\mu >0$, length of scouting phase $\Tint$, discretization size $\epsilon$, confidence $\delta>0$.
\State \textbf{Initialize}: $\cTpar = \cTint = \varnothing$, $\bV = \mathbf{I}_{d}$, $K = \lceil 2P/\epsilon\rceil + 3$, $\cK = \range{-K,K}$, $\widetilde\epsilon=(12PL + 7)\epsilon$, \\ $N^{k,s} = N^{k,b} = \widehat{F}^{k} = \widehat{D}^{k} = 0 $ for all $k \in \cK$. Let $\beta:\mathbb{N}\ni n\mapsto \sqrt{2\log(\delta^{-1})/n}$
\While{$t \leq T$}
\If{$\Vert x_t\Vert_{\bV^{-1}} > \mu$}

\State $\cTpar\gets \cTpar\cup\{t\}$, $\textsc{Est-Par}(\cTpar,\{s,b\})$

\ElsIf{$\vert \cTint \vert < \Tint$} 
\State $\cTint\gets \cTint\cup\{t\}$, $\textsc{Est-Int}(\cTint,\{s,b\})$
\Else{} \Comment{Run Successive Elimination}

 \For{$(k, k') \in \cA_t \defeq \mleft\{(k, k') \in \cK^2: k' = \mleft\lfloor \left(x_t^{\top} \left(\widehat{\theta}^s_t - \widehat{\theta}^b_t\right) + k\epsilon\right)\epsilon^{-1} \mright\rfloor\mright\}$}
\State $\UCB_t(k,k') \gets \widehat{I}^{k'}\widehat{F}_t^k+\widehat{J}^{k} \widehat{D}_t^{k'} + \widetilde\epsilon + 2P\left(\beta(N^{k,s}) + \beta(N^{k',b})\right)$


\State $\LCB_t(k,k') \gets \widehat{I}^{k'}\widehat{F}_t^k+\widehat{J}^{k} \widehat{D}_t^{k'} - \widetilde\epsilon - 2P\left(\beta(N^{k,s}) + \beta(N^{k',b})\right)$
\EndFor
\State $\cK_t  \defeq  \mleft\{(k, k')\in \cA_t : \UCB_t(k,k') \geq \max_{(l, l') \in \cA_t}\LCB_t(l,l')\mright\}$\label{algline:k successive elimination}
\State Choose $(k_t, k_t') \in \argmin_{(k, k') \in \cK_t}\,\min\{N^{k,s}, N^{k',b}\}$ and post price $p_t = x_t^{\top}\widehat{\theta}^s_t + k_t\epsilon$ \label{algline:k min}
\State Observe feedback and update: 
\vspace{-2mm}\[\widehat{F}^{k_t} \gets  \frac{N^{k_t, s}\widehat{F}^{k_t}+\mathbb{I}\left\{s_t \leq p_t\right\}}{N^{k_t, s} + 1} \quad \textnormal{ and }\quad \widehat{D}^{k'_t} \gets  \frac{N^{k'_t, b}\widehat{D}^{k'_t}+\mathbb{I}\left\{b_t \geq p_t\right\}}{N^{k'_t, b} + 1},\]\vspace{-2mm}
\State $ \ \quad N^{k_t, s} \gets  N^{k_t, s} + 1  \quad \textnormal{ and }\quad N^{k'_t, b} \gets  N^{k'_t, b} + 1$
\EndIf
\EndWhile
\end{algorithmic}
\end{algorithm}

\begin{theorem}\label{thm:noise_2_bits}
Assume that the valuations of the buyer and the seller follow the noisy linear model in Equation \eqref{eq:2bit_noise}. Let $\epsilon = (d^2\log(T)^2/T)^{\nicefrac{1}{3}}$, 
$\delta = \left((38 + 16P\epsilon^{-1})(T+1)^2\right)^{-1}$, $\mu = \epsilon \,\big(P\sqrt{d\log\left(\frac{1 + B^2T}{\delta}\right)} + A\big)^{-1}$, and $\Tint = \lceil 8P^2\log(1/\delta)/\epsilon^2\rceil$. Then, under Assumptions \ref{ass:bounded_context} and \ref{ass:bounded_densities}, with probability at least $1 - 1/T$, \Cref{alg:SBLBT} has a regret bounded by
$
 R_T \leq \widetilde O(T^{\nicefrac{2}{3}}).
$
\end{theorem}

\begin{sproof} 
Let $\sr(x, p) \defeq \max_{p'}\egft(x,p')- \egft(x,p)$, and $\Tpar = \vert \cTpar\vert$.
The regret is bounded by $2P\Tpar + 2P\Tint + \sum_{t \notin \cTpar \cup \cTint} \sr(x_t, p_t)$. The elliptical potential Lemma \citep{carpentier2020elliptical} implies that $\Tpar \lesssim d^2\log(T)^2\epsilon^{-2}$, so $\Tpar + \Tint \lesssim d^2\log(T)^2\epsilon^{-2}$. 
By using concentration arguments and exploiting the Lipschitz continuity of the reward, we can show that, for all $t\notin \cTpar\cup\cTint$ and $(k, k') \in \cA_t$, $[\LCB_t(k, k'),\UCB_t(k, k')]$ is a valid confidence interval for $\egft(x_t, x_t^{\top}\widehat{\theta}^s + k\epsilon)$. 
Now, we want to show that $\sr(x_t, p_t) \lesssim\widetilde O\big(\sqrt{\nicefrac{1}{N^{k_t,s}}}+ \sqrt{\nicefrac{1}{N^{k'_t,b}}} + \epsilon\big)$.
To prove this claim, we define $(k^*_t, k_t^{'*}) \in \argmax_{(k, k')\in \cA_t} \egft(x_t,x_t^{\top}\widehat{\theta}^s + k\epsilon)$. The Lipschitz-continuity of the reward ensures that the discretization error
is of order $O(\epsilon)$. Moreover, $[\LCB_t(k_t, k'_t),\UCB_t(k_t, k'_t)]$ and $[\LCB_t(k^*_t, k_t^{'*}),\UCB_t(k^*_t, k_t^{'*})]$ are confidence intervals for $\egft(x_t, p_t)$ and $\egft(x_t, x_t^{\top}\widehat{\theta^s} + k_t^*\epsilon)$, respectively, and $\LCB_t(k^*_t, k_t^{'*})\leq \UCB_t(k_t, k'_t)$ by Line \ref{algline:k successive elimination}.
Then, it holds that $\egft(x_t, x_t^{\top}\widehat{\theta^s} + k_t^*\epsilon) - \egft(x_t, p_t)$ is of order $\UCB_t(k_t, k'_t)- \LCB_t(k_t, k'_t) + \UCB_t(k^*_t, k_t^{'*}) - \LCB_t(k^*_t, k_t^{'*})$. Our choice of $(k_t, k'_t)$ ensures that $\min\{N^{k_t,s}, N^{k'_t, b}\} \leq \min\{ N^{k^*_t, s}, N^{k_t^{'*}, s}\}$ (Line \ref{algline:k min}).
Then, by definition of $\UCB$ and $\LCB$, and since the discretization is at most $O(\epsilon)$, we get $\sr(x_t, p_t) \lesssim \widetilde O\big(\sqrt{\nicefrac{1}{N^{k_t,s}}}+ \sqrt{\nicefrac{1}{N^{k'_t,b}}} + \epsilon\big)$. 

To bound the regret of the Successive Elimination phase, we consider separately the rounds (indexed by $\cT^s_{k}$) where this bound is dominated by $\sqrt{\nicefrac{1}{N^{k,s}}}$ and $k_t = k$, and the rounds (indexed by $\cT^b_{k'}$) where it is dominated by $\sqrt{\nicefrac{1}{N^{k',b}}}$ and $k_t' = k'$. Then, we decompose the regret of the Successive Elimination phase as $\sum_{t \notin \cTpar \cup \cTint} \sr(x_t, p_t) = \sum_{k \in \cK} (\sum_{t \in \cT^s_k} \sr(x_t, p_t) + \sum_{t \in \cT^b_k} \sr(x_t, p_t)).$ 
Choosing $\sr_a = 2^{-a}$ and $\sr_{\overline{a}} \approx \epsilon$, we consider the decreasing sequence of intervals $(\sr_a)_{a \leq \overline{a}}$. The previous result implies that if $N^{k, s} \in [\widetilde O(\nicefrac{1}{\sr_a^2}), \widetilde O(\nicefrac{1}{\sr_{a+1}^2})]$, then $\sr(x_t, p_t) \lesssim \sr_a$. This implies that
%
%
\[
\sum_{t \in \cT^s_k} \sr(x_t, p_t) \lesssim \widetilde O\bigg(\sr_{1}^{2} + \sum_{1 \leq a \leq  \overline{a}-1}\sr_a\left(\sr_{a+1}^{2} - \sr_{a}^{2}\right) + \sr_{\overline{a}} \vert \cT_k^s\vert\bigg).
\]
This sum is of order $\T^s_k\epsilon + \nicefrac{\log(\nicefrac{1}{\delta})}{\epsilon}$. We sum over $(k, k') \in \cK$, and use the fact that $\sum_{k\in \cK} \T^s_k+\T^b_k \leq T$ while $\vert \cK\vert \approx \epsilon^{-1}$, to conclude that $R_T=\widetilde O(T^{\nicefrac23})$.
\end{sproof}

\begin{remark}
\citet{CesaBianchiCCF21} studied the related problem of non-contextual stochastic bilateral trade under similar assumptions (independent valuations with bounded densities and support). They obtain a lower bound on the regret of order $\Omega(T^{\nicefrac23})$. This shows that, up to the constant $\widetilde C'$, the dependence on the dimension $d$, and logarithmic factors, \Cref{alg:SBLBT} enjoys minimax optimal regret.
\end{remark}
\section{Trading off budget balance constraints for feedback}\label{sec:one bit}

In this section, we show that it is possible to relax the budget balance constraints in order to handle scenarios with limited feedback. In particular, we show that any EOC-like algorithm for two-bit feedback (for instance, \Cref{alg:EPBT,alg:explore-commit}) can be suitably adapted to handle settings with one-bit feedback. To achieve this, we resort to the notion of \textit{global budget balance} recently introduced by \citet{bernasconi2023no} and we show that budget can be exploited to compensate for the lack of feedback.
The notion of global budget balance requires the learner to be budget balanced only ``overall'' (\ie over the whole time horizon). In particular, let $\prof_t(p,q)\defeq \indicator{s_t\le p}\indicator{q\le b_t}(q-p)$ be the profit extracted by the learner at time $t$ by posting prices $(p,q)$. Then, the learning algorithm is global budget balanced if the following inequality holds: $\sum_{t=1}^T\prof(p_t,q_t)\ge 0$.

Given an EOC-like algorithm for the two-bit strong budget balance set-up, the idea is to simulate its exploration rounds by doing separate updates for the seller and the buyer, each using a single bit of feedback. In particular, each time in which the two-bit algorithm uses the seller's feedback $\indicator{s_t\le p}$, the one-bit algorithm needs to actively collect that information by posting $(p,-P)$ instead of $(p,p)$. 
Notice that, even under one-bit feedback, if the learner posts price $(p,-P)$ it is able to observe $\indicator{s_t \le p}$ since the buyer is always going to accept the trade. 
Analogously, when the two-bit algorithm would employ $\indicator{b_t\ge p}$, the one-bit algorithm has to collect that information by posting $(P,p)$. 
For instance, instead of executing $\textsc{Est-Par}$, the one-bit algorithm runs two exploration phases: one in which it updates only $\widehat\theta^s$ (Line \ref{algline:update theta s} of $\textsc{Est-Par}$) by posting $(p_t,-P)$, and one in which it updates only $\widehat\theta^b$ by posting $(P,p_{t'})$. The same happens for the other estimation subroutines.

The one-bit algorithm simulates access to finer-grained feedback by posting pairs of prices $(P,p)$, $(p,-P)$. These prices result in a negative GFT. By enforcing global budget balance, we allow the learning algorithm to set prices that are not budget balanced individually, as long as the overall budget balance is maintained. To do that, we add an initial ``budget-collection'' phase where the learning algorithm builds up enough profit to be able to sustain the following exploration round. 
We demonstrate that sufficient profit can be accumulated during the budget-collection phase without incurring excessive regret relative to GFT. This allows us to provide sublinear regret guarantees under one-bit feedback, with rates depending on some natural characteristics of the underlying market.


\xhdr{Collecting budget.} Let $\budget$ be the budget required by the one-bit algorithm to cover the rounds in which it must post either $(P, p)$ or $(p, -P)$ to collect information (the value of $\budget$ will be set later). Therefore, the budget collection phase ends as soon as $\sum_{t}\prof(p_t,q_t)\ge \budget$. We denote by $\tau$ the random variable indicating the last round of the budget-collection phase. Moreover, let $\cTB\defeq\range{\tau}$ be the random set of rounds designated for budget collection.
For each round $t\in\cTB$, the algorithm samples a price $p_t$ uniformly from $[-P,P]$, and a value $i_t$ uniformly from $[0,\log T]$. Then, it posts a pair of prices $(p_t,q_t)$, where $q_t=p_t+2^{-i_t}$.
We start by providing a lower bound to the per-round expected profit obtained by posting $(p_t,p_t+2^{-i_t})$.

\begin{restatable}{lemma}{profitLemma}\label{lemma:profit}
For each round $t\in\cTB$ such that $b_t\ge s_t$, it holds:
$
\E{\prof_t(p_t,q_t)} \ge  \frac{(b_t-s_t)^2}{8P\log T}-\frac{2}{T},
$
where the expectation is with respect to the choice of $(p_t,q_t)$.
\end{restatable}

The guarantees provided by the one-bit algorithm depend on how ``active'' is the market. We start by assuming that the market is well-behaved, meaning that trade opportunities arise frequently enough. Formally, we require that for a sufficiently long time  $t'$, the cumulative social welfare is $\Omega(t')$.

\begin{assumption}\label{assumption:active market}
Given $\alpha\in(0,1)$, for each $t'\ge\log(T)$, it holds $\sum_{t \in \range{t'}}(b_t-s_t) \ge \alpha t'$.\footnote{The case in which the condition holds with high probability can be addressed with minor variations. }
\end{assumption}

\Cref{remark: one bit no assumption} described the guarantees that can be obtained in the absence of this assumption.
We observe that the budget-collection phase will have at least length $\log T$ since $\tau\in[\nicefrac{\budget}{2P},T]$, and $\budget=\Omega( \log T)$ even in the setting without noise. Then, assuming the environment allows for enough budget to be collected during the first phase, we can show that the profit accumulated by time $\tau$ is at least a factor of $1/\log T$ of the cumulative GFT up to that point, assuming all trades up to $\tau$ had occurred.
\begin{restatable}{lemma}{logTgft}\label{lemma:logT gft}
    Under \Cref{assumption:active market}, for $\tau\ge\log T$, it holds with probability at least $1-1/T$ that
    \[ \sum_{t \in \range{\tau}} \prof_t(p_t,q_t) \ge
    \frac{\alpha}{8 P \log(T)}  \sum_{t\in \range{\tau}} [b_t-s_t]^+ - \sqrt{4 P^2 \log(T)\sum_{t \in \range{\tau}} [b_t-s_t]^+ } -2.\]
\end{restatable}
Given a two-bit algorithm $\algA$ such that, with probability at  least $1-1/T$, it requires at most  $\te$ exploration rounds, and has regret for the commit phase of $\regtwo$, by setting the overall budget to  $\budget=\max\mleft\{2048P^4\alpha^{-2}\log^3T,\, 2P\te\mright\}$, the one-bit learning algorithm has the following guarantees.\footnote{We observe that, with probability at most $1/T$, the budget $\budget$ is not sufficient to cover $\te$ exploration rounds. When this occurs, we stop the exploration and start posting arbitrary strong budget balanced prices.} 

\begin{restatable}{theorem}{oneBitThm}\label{thm: one bit regret}
Under \Cref{assumption:active market}, given the two-bit algorithm $\algA$, the corresponding one-bit learning algorithm satisfies global budget balance and, with probability at least $1-1/T$, has regret \[\regone\le  O(\te)\cdot \textnormal{polylog}(T) +\regtwo.\]
\end{restatable}

For instance, under the set-up of \Cref{sec:noisy}, we can use \Cref{alg:explore-commit} as $\algA$. Then, $\budget\le 4P\te = 4P \cdot 2P\left(\vert\cTpar_{T+1} \vert + \Tint + 2\vert \cK\vert \TFD\right)$, which is of order $\widetilde O(P^4 d T^{\nicefrac34})$. Therefore, the final regret bound in the noisy, one-bit setting with global budget balance is of order $\widetilde O(T^{3/4})$. A similar argument shows that in the noiseless setting of \Cref{sec:deterministic} the regret is of order $\textnormal{polylog}(T)$.

\begin{remark}\label{remark: one bit no assumption}
In the absence of \Cref{assumption:active market}, we can still establish sublinear regret guarantees under one-bit feedback and global budget balance. The main modification to the previous analysis is halting the proof of \Cref{lemma:logT gft} prior to employing \Cref{assumption:active market}. This results in the following bound:
\[ \textstyle  \sum_{t \in \range{\tau}}  \prof_t(p_t,q_t) \ge \Omega\mleft( \tau^{-1}\mleft(\sum_{t \in \range{\tau}}([b_t-s_t)]^+\mright)^2  \mright)-\widetilde O(\sqrt{T}).  \] 
By selecting an appropriate duration for the exploration phase of the two-bit algorithm, we can employ this inequality to bound the regret during the budget-collection phase. For example, by setting $\cte\approx \log T$, we achieve a regret of $\widetilde O(T^{3/4})$ in the deterministic set-up of \Cref{sec:deterministic}, and by setting $\te\approx T^{\nicefrac34}$ we achieve a regret of $\widetilde O(T^{\nicefrac78})$ in the noisy set-up of \Cref{sec:noisy}.
\end{remark}

\clearpage
\section*{Acknowledgments}

MB, MC, and AC are partially supported by the FAIR (Future Artificial Intelligence Research) project PE0000013, funded by the NextGenerationEU program within the PNRR-PE-AI scheme (M4C2, investment 1.3, line on Artificial Intelligence). MC is also partially supported by the EU Horizon project ELIAS (European Lighthouse of AI for Sustainability, No. 101120237). AC is partially supported by MUR - PRIN 2022 project 2022R45NBB funded by the NextGenerationEU program. VP acknowledges support from the French National Research Agency (ANR) under grant number ANR-19-CE23- 0026 as well as from the grant “Investissements d’Avenir” (LabEx Ecodec/ANR-11-LABX-0047). SG gratefully acknowledges funding from the Hadamard Doctoral School of Mathematics (EDMH).

\bibliographystyle{plainnat}
\bibliography{refBT}

\clearpage
\appendix

\section{Basic facts about ellipsoids}\label{app:ellipsoid}

An ellipsoid $E$ with center $c\in\R^d$ and positive definite matrix $\bM\in\R^{d\times d}$ is defined as 
\[
E(\bM,c)\defeq \mleft\{x\in\R^d: \|x-c\|_{ \bM^{-1}}\le 1\mright\}.
\]
Given an hyperplane $H(a,c)\defeq\{x\in\R^d:a^\top(x-c)\le 0\}$, the minimal volume ellipsoid containing $K=E(\bM,c)\cap H(a,c)$ can be computed in closed form as follows. The new center is
\[
c'=c-\frac{1}{d+1}\cdot\frac{1}{\sqrt{a^\top\bM a}}\bM a,
\]
and the new positive definite matrix is 
\[
\bM'=\frac{d^2}{d^2-1}\mleft(\bM - \frac{2}{d+1}\cdot\frac{\bM aa^\top\bM}{a^\top\bM a}\mright).
\]
Then, the L\"owner-John ellipsoid $\LJ(K)$ of the half ellipsoid $K$ is the ellipsoid $E(\bM',c')$.

Going from $E(\bM,c)$ to $E(\bM',c')$ the volume shrinks by the following amount:
\[
\textnormal{vol}(E(\bM',c'))\le e^{\nicefrac{-1}{2d}}\textnormal{vol}(E(\bM,c)).
\]
An in-depth analysis of these results can be found in \citet{grotschel2012geometric}.
\section{Proof of Theorem \ref{thm:no_noise_2_bits}}\label{app:2bitnonoise}

\noNoiseEllipsoid*

\begin{proof}
The proof of \Cref{thm:no_noise_2_bits} relies on the fact that when the algorithm updates $E^{s}_{t}$ or $E^{b}_{t}$, it behaves as the \textsc{EllipsoidPricing} algorithm of \citet{Cohen_feature_DP}. More precisely, we first note that under \Cref{ass:bounded_context}, $E^s_1$ and $E^b_1$ contain $\theta^s$ and $\theta^b$, respectively. Moreover, straightforward induction shows that if $E^s_t$ contains $\theta^s$, then $K^s_{t+1}$ also contains $\theta^s$, and so does $E^s_{t+1}$ (similar reasoning applies to $\theta^b$ and $E^b_t$). 

This implies that for all rounds, $\underline{s_t} \leq s_t \leq \overline{s_t}$, and $\underline{b_t} \leq b_t \leq \overline{b_t}$. In particular, if the condition in Line \ref{algline:exploit} is verified, we have $s_t \leq p_t \leq b_t$, and so the instantaneous regret at that round is 0.

We also underline that \Cref{alg:EPBT} only explores to estimate $\theta^s$ (Line \ref{algline:explore_s}) if the condition 
$$\max_{\theta \in E^s_t}x_t^{\top}\theta - \min_{\theta \in E^s_t}x_t^{\top}\theta \geq \epsilon$$
is verified, or equivalently if
$$\max_{\theta \in E^s_t}\frac{x_t^{\top}}{B}\theta - \min_{\theta \in E^s_t}\frac{x_t^{\top}}{B}\theta \geq \frac{\epsilon}{B}.$$
Since the normalized contexts $\frac{x_t^{\top}}{B}$ are bounded in norm by 1 by \Cref{ass:bounded_context}, we can apply Lemma 1 by \citet{Cohen_feature_DP}. This Lemma states that \Cref{alg:EPBT} will execute the steps in Lines \ref{algline:explore_s_beg}-\ref{algline:explore_s_end} at most $2d^2\log(20A(d+ 1)/\frac{\epsilon}{B})$ times, after which the condition in Line \ref{algline:explore_s} will never be verified again. A similar reasoning allows us to bound the number of exploration rounds for the buyer's parameter $\theta^b$ (Line \ref{algline:explore_b}). Therefore, the total number of exploration rounds is bounded by $2\times 2d^2\log(20A(d+ 1)/\frac{\epsilon}{B})$. We notice that \Cref{ass:bounded_context} also implies that the valuations are in $[-AB, AB]$, and so the instantaneous regret is bounded by $2AB$. Thus, the regret of the exploration phase (Lines \ref{algline:explore_s} and \ref{algline:explore_b}) is bounded by  $8ABd^2\log(20A(d+ 1)/\frac{\epsilon}{B})$.

Finally, if the conditions on Lines \ref{algline:exploit},  \ref{algline:explore_s} and \ref{algline:explore_b} are not verified, it implies that 
\begin{align*}
    b_t - s_t &\leq \overline{b}_t - \underline{s_t}\\
     &\leq \underline{b}_t + \epsilon - \overline{s_t} + \epsilon\\
     &\leq 2\epsilon.
\end{align*}
Then, the instantaneous regret at that round is at most $2\epsilon$, and the total regret for these steps is at most $2\epsilon T$. Choosing $\epsilon = \frac{ABd^2}{T}$ yields the result.
\end{proof}
\section{Proof of \Cref{thm:etc}}\label{app:explore_commit}
In order to prove Theorem \ref{thm:etc}, we provide a more detailed version of the \textsc{Explore-Or-Commit} algorithm in Algorithm \ref{alg:ETC_detailed}. Note that we adopt the convention $1/0 = +\infty$.


\begin{algorithm}
\caption{\textsc{Explore-Or-Commit} (detailed version)}\label{alg:ETC_detailed}
\begin{algorithmic}[1]
\State \textbf{Input}: parameter $\mu >0$, length of exploration phases $\Tint$ and $\TFD$, discretization size $\epsilon$, confidence level $\delta$.
\State \textbf{Initialize}: $\cTpar_1 = \cTint = \varnothing$, $\widehat{\theta}_1^s = \mathbf{0}_d$, $\widehat{\theta}_1^b = \mathbf{0}_d$, $\bV_1 = \mathbf{I}_d$, $K = \lceil 2P/\epsilon\rceil + 3$, $\cK = \llbracket -K,K\rrbracket$, $\cTF_k = \cTD_k = \varnothing$ for all $k\in\cK$, $\widehat{F}^k =\widehat{D}^k = 0$ for all $k\in\cK$.

\vspace{.5em}
\While{$t \leq T$}
\If{$\Vert x_t\Vert_{\bV_t^{-1}} > \mu$} \Comment{Estimate the parameters $\theta^b$ and $\theta^s$}
    \State Draw $p_t \sim \cU([-P,P])$ and post prices $(p_t,p_t)$
    \State Update $\cTpar_{t+1} \gets \cTpar_t \cup \{t\}$
    \State Update $\bV_{t+1} =  \left(\underset{l \in \cTpar_{t+1}}{\sum} x_l x_l^{\top} +  \bI_d \right)^{-1}$
    \State Update parameter estimates \[\widehat{\theta}^s_{t+1} = 2P\bV_{t+1}\underset{l \in \cTpar_{t+1}}{\sum}\left(\mathbb{I}\{p_l\leq s_l\} - \frac{1}{2}\right)x_l,\quad \widehat{\theta}^b_{t+1} = 2P\bV_{t+1}\underset{l \in \cTpar_{t+1}}{\sum}\left(\mathbb{I}\{p_l \leq b_l\}-\frac{1}{2}\right)x_l\]\label{line:est par detailed}
\ElsIf{$\vert \cTint \vert < \Tint$} \Comment{Estimate the integrals $I$ and $J$}
\State Draw $p_t \sim \cU([-P,P])$ and post prices $(p_t,p_t)$
\State Update $\cTint \gets \cTint \cup \{ t \}$
\If{$\vert \cTint \vert = \Tint$}
\State Compute estimate of integrals $I$ and $J$\label{line:est int detailed}
\For{$k\in \cK$}
\[\widehat{I}^k = \frac{2P}{\Tint}\sum_{l \in \cTint}\mathbb{I}\left\{k\epsilon + x_l^{\top}\widehat{\theta}^b \leq p_l \leq b_l  \right\},\quad \widehat{J}^k = \frac{2P}{\Tint}\sum_{l \in \cTint}\mathbb{I}\left\{s_l \leq p_l \leq k\epsilon + x_l^{\top}\widehat{\theta}^s \right\}.\]
\EndFor
\EndIf
\ElsIf{for some $k\in \cK$, $\vert \cTF_k \vert < \TFD$}\Comment{Estimate $F_k$}
\State Set $p_t = x_t^{\top}\widehat{\theta}^s + k\epsilon$ and post $(p_t,p_t)$
\State Update $\cTF_k \gets \cTF_k \cup \{ t \}$
\If{$\vert  \cTF_k \vert = \TFD$} set \[\textstyle{\widehat{F}^k = \frac{2P}{\TFD}\sum_{l \in \cTF_k}\mathbb{I}\left\{s_l \leq p_l\right\}}\]\EndIf
\ElsIf{for some $k\in \cK$, $\vert \cTD_k \vert < \TFD$}\Comment{Estimate $F_k$}
\State Set $p_t = x_t^{\top}\widehat{\theta}^b + k\epsilon$ and post $(p_t,p_t)$ 
\State Update $\cTD_k \gets \cTD_k \cup \{ t \}$
\If{$\vert  \cTD_k \vert = \TFD$}  set \[\textstyle{\widehat{D}^k = \frac{2P}{\TFD}\sum_{l \in \cTD_k}\mathbb{I}\left\{p_l\leq b_l\right\}}\]
 \EndIf
\Else \Comment{Post greedy price}
\State set $\cA_t = \Bigg\{(k, k') \in \cK^2 : k' = \left\lfloor \frac{x_t^{\top} \left(\widehat{\theta}^s_t - \widehat{\theta}^b_t\right) + k\epsilon}{\epsilon} \right\rfloor\Bigg\}$
\State $k_t = \argmax_{k : \exists k', (k,k') \in \cA_t} \widehat{I}^{k'}\widehat{F}^k+\widehat{J}^{k} \widehat{D}^{k'}$.
\State Set $p_t = x_t^{\top}\widehat{\theta}^s_t + k_t\epsilon$ and post $(p_t,p_t)$
\EndIf
\EndWhile
\end{algorithmic}
\end{algorithm}

Before analyzing the \textsc{Explore-Or-Commit} algorithm, we define the various quantities used within it. Specifically, $\cTpar_t$ represents the set of indices from previous rounds that were dedicated to estimating the parameters $\theta^s$ and $\theta^b$ (Line \ref{line:est par detailed}).
We denote by $\widehat{\theta}^s_t$ and $\widehat{\theta}^b_t$ the estimates for the parameters, and $\bV_t$ is the corresponding empirical covariance matrix at the beginning of round $t$. 
The set $\cTint$ contains the indices of the rounds used to estimate the integrals $I$ and $J$ (Line \ref{line:est int detailed}). Moreover, $\widehat{I}^k$ is the estimate of  $I(k\epsilon)$, and $\widehat{J}^k$ is the estimate of $J(k\epsilon)$. The quantities $\widehat{F}^k$ and $\widehat{D}^k$ represent our estimates of $F^s(k\epsilon)$ and $D^b(k\epsilon)$, respectively.


We use the following lemma to bound the total duration of the parameter exploration phase (the proof can be found in \Cref{app:proof lemma stop tpar}). 
\begin{restatable}{lemma}{stopTpar}\label{lem:explore_theta}
    Almost surely, the length of exploration phase $\cTpar_{T+1}$ is bounded as 
    $$\vert\cTpar_{T+1}\vert \leq \frac{d \log\left(\frac{T + d}{d}\right)}{\mu^2}.$$
\end{restatable}


The following lemma provides bounds on the estimation errors for $\theta^s$, $\theta^b$, $I$, $J$, $F^s$, and $D^b$ for the values of $\mu$, $\Tint$, and $\TFD$ specified in Theorem \ref{thm:noise_2_bits}.
Let $\cE$ be the event
\begin{align*}
\cE \defeq \Big\{ \forall k \in \cK, \ \forall t \notin \cTpar_{T}\cup \cTint\,\,& \left\vert x_t^{\top}(\widehat{\theta}^b_t- \theta^b) \right\vert \leq \epsilon \textnormal{ and }\\  &\left\vert x_t^{\top}(\widehat{\theta}^s_t- \theta^s) \right\vert \leq \epsilon \textnormal{ and }\\ 
&\left\vert\widehat{I}^k - I\left(k\epsilon\right) \right\vert \leq 2\epsilon \textnormal{ and }\\  &\left\vert\widehat{J}^k - J(k\epsilon) \right\vert \leq 2\epsilon\Big\}.
\end{align*}

Let us define the exploitation (\ie commit) phase as
\[\cTC \defeq \llbracket T \rrbracket \setminus \left(\cTpar_{T}\cup \cTint \underset{k\in\cK}{\bigcup}\mleft( \cTF_k\cup\cTD_k\mright)\right).\]
Moreover, let $\cE^{\textnormal{EOC}}$ be the event
\begin{eqnarray*}
\cE^{\textnormal{EOC}} \defeq \cE \cap \mleft\{\forall k \in \cK, \ \forall t \in \cTC, \left\vert \widehat{F}^{k} - F^s(k\epsilon) \right\vert \leq  (L+1)\epsilon \text{ and }  \left\vert \widehat{D}^{k} - D^s(k\epsilon) \right\vert \leq  (L+1)\epsilon \mright\}.
\end{eqnarray*} 
Then, we can lower bound the probability of event $\cE^{\textnormal{EOC}}$ as follows (the proof is postponed to \Cref{app:bound event etc lemma}).

\begin{restatable}{lemma}{boundetcevent}\label{lem:bound_epsilon_etc}
For the choice $\mu = \epsilon\left(P\sqrt{d\log\left(\frac{1 + B^2T}{\delta}\right)} + A\right)^{-1}$, $\Tint = 8P^2\log(1/\delta)/\epsilon^2$, and $\TFD = 2\log(1/\delta)\epsilon^{-2}$, it holds that
$$\mathbb{P}\left(\cE^{\textnormal{EOC}}\right)\geq 1-2\delta - 8\delta(2K+1).$$
\end{restatable}
Note that the choice $\delta = \left(T(74+32P\epsilon^{-1})\right)^{-1}$ ensures that the event $\cE^{\textnormal{EOC}}$ happens with probability at least $1-T^{-1}$.

The following Lemma bounds the error for estimating the gain from trade of a price $p = x_t^{\top}\widehat{\theta}^s_t + k\epsilon$ on the high-probability event $\cE^{\textnormal{EOC}}$ (the proof is presented in \Cref{app:proof conf bound etc}).

\begin{restatable}{lemma}{errgftevent}\label{lem:conf_bound_etc}
On the event $\cE^{\textnormal{EOC}}$, for all $t \in \cTC$, and all $(k,k')\in \cA_t$,
$$\left \vert \widehat{I}^{k'}\widehat{F}^k+\widehat{J}^{k} \widehat{D}^{k'} - \egft(x_t, x_t^{\top}\widehat{\theta}^s_t + k\epsilon)\right \vert \leq  (10PL + 4P + 7)\epsilon.$$
\end{restatable}
Next, we bound the discretization error. Let us define
$$(k^*_t, k'^*_t) \in \argmax_{(k, k') \in \cA_t} \egft(x_t, x_t^{\top}\widehat{\theta}_t^s + k\epsilon).$$
Then, the following result holds (see \Cref{app:proof discretization error} for its proof).
\begin{restatable}{lemma}{errdiscrevent}\label{lem:discretisation_error}
On the event $\cE^{\textnormal{EOC}}$, we have that
$$\left \vert \max_p\egft(x_t, p) - \egft(x_t, x_t^{\top}\widehat{\theta}_t^s + k_t^*\epsilon)\right \vert \leq 2LP\epsilon.$$
\end{restatable}

We are now ready to prove Theorem \ref{thm:etc}. For $p \in \R$, $x \in \mathbb{R}^d$, we define 
$$\Delta(x, p) \defeq \max_{p'} \egft(x,p') - \egft(x,p).$$ 
We begin by decomposing the regret as 
\begin{align*}
    R_T =& \sum_{t\in \cTpar_{T+1}}\Delta(x_t, p_t) + \sum_{t\in \cTint}\Delta(x_t, p_t) +\\& \sum_{k\in \cK}\left(\sum_{t\in \cTF_k}\Delta(x_t, p_t) + \sum_{t\in \cTD_k}\Delta(x_t, p_t)\right) + \sum_{t\in \cTC}\Delta(x_t, p_t).
\end{align*}
Using the fact that $\Delta(x_t, p_t) \leq 2P$, we obtain 
\begin{align*}
    R_T \leq& 2P\left(\vert\cTpar_{T+1} \vert + \Tint + 2\vert \cK\vert \TFD\right) + \sum_{t\in \cTC}\Delta(x_t, p_t).
\end{align*}
On the one hand, Lemma \ref{lem:explore_theta} implies that $\vert\cTpar_{T+1}\vert \leq \frac{d \log\left(\frac{T + d}{d}\right)}{\mu^2}.$ For the choice  $\mu = \epsilon\left(P\mleft(d\log\left(\frac{1 + B^2T}{\delta}\right)\mright)^{\nicefrac12} + A\right)^{-1}$, $\Tint = 8P^2\log(1/\delta)\epsilon^{-2}$, and $\TFD = 2\log(1/\delta)\epsilon^{-2}$, this implies that
\begin{align*}
    R_T \leq& 2Pd \log\left(\frac{T + d}{d}\right)\left(P\sqrt{d\log\left(\frac{1 + B^2T}{\delta}\right)} + A\right)^2\epsilon^{-2} + 16P^3\log(1/\delta)\epsilon^{-2} +\\
    &8P\left(\frac{4P}{\epsilon}+9\right)\log(1/\delta)\epsilon^{-2} + \sum_{t\in \cTC}\Delta(x_t, p_t).
\end{align*}
Note that the first term is of order $\epsilon^{-3}$ when $\epsilon$ is small enough. On the other hand, on the event $\cE^{\textnormal{EOC}}$, by Lemma \ref{lem:discretisation_error} we have
\begin{align*}
    &\hspace{-.5cm}\sum_{t\in \cTC}\Delta(x_t, p_t) \\&= \sum_{t\in \cTC}\Delta(x_t, x_t^{\top}\widehat{\theta}^s_t + k^*_t\epsilon) + \sum_{t\in \cTC}\egft(x_t, x_t^{\top}\widehat{\theta}_t^s + k_t^*\epsilon) - \egft(x_t, x_t^{\top}\widehat{\theta}_t^s + k_t\epsilon)\\
    &\leq 2TLP\epsilon + \sum_{t\in \cTC}\mleft(\egft(x_t, x_t^{\top}\widehat{\theta}_t^s + k_t^*\epsilon) - \egft(x_t, x_t^{\top}\widehat{\theta}_t^s + k_t\epsilon)\mright),
\end{align*}
Moreover, our choice of $k_t$ ensures that for $k'_t$ such that $(k_t, k'_t)\in \cA_t$,
\begin{align*}
\widehat{I}^{k'_t}\widehat{F}^{k_t}+\widehat{J}^{k_t} \widehat{D}^{k'_t}&\geq \widehat{I}^{k_t^{'*}}\widehat{F}^{k^*_t}+\widehat{J}^{k^*_t} \widehat{D}^{k_t^{'*}}.
\end{align*}
By Lemma \ref{lem:conf_bound_etc}, this implies that, on event $\cE^{\textnormal{EOC}}$,
\begin{align*}
\egft(x_t, x_t^{\top}\widehat{\theta}_t^s + k_t^*\epsilon) - \egft(x_t, x_t^{\top}\widehat{\theta}_t^s + k_t\epsilon)\leq (20PL + 8P + 14)\epsilon.
\end{align*}

Thus, on the event  $\cE^{\textnormal{EOC}}$,
\begin{align*}
    R_T \leq& 2Pd \log\left(\frac{T + d}{d}\right)\left(P\sqrt{d\log\left(\frac{1 + B^2T}{\delta}\right)} + A\right)^2\epsilon^{-2} + 16P^3\log(1/\delta)/\epsilon^2 \\
    &+ 8P\left(\frac{P}{\epsilon}+3\right)\log(1/\delta)\epsilon^{-2} + 2TLP\epsilon + T(20PL + 8P + 14)\epsilon.
\end{align*}
Substituting the values of $\epsilon$ and $\delta$ as per Theorem \ref{thm:etc} and using Lemma \ref{lem:bound_epsilon_etc} allows us to conclude the proof.

\section{Proof of Theorem \ref{thm:noise_2_bits}}\label{app:2bit noise scouting}

\subsection{Detailed Algorithm}

In order to prove Theorem \ref{thm:noise_2_bits}, we provide a more detailed version of the \textsc{Scouting Bandit with Information Pooling} algorithm in Algorithm \ref{alg:SBLBT_detailed}. Note that in Lines 16 an 17, we adopt the convention $1/0 = +\infty$ when computing the upper- and lower- confidence bounds for increment levels $k, k'$ that have not yet been selected.


\begin{algorithm}
\caption{\textsc{Scouting Bandits with Information Pooling (SBIP)} (detailed version)}\label{alg:SBLBT_detailed}
\begin{algorithmic}[1]
\State \textbf{Input}: parameter $\mu >0$, length of scouting phase $\Tint$, discretization size $\epsilon$, confidence level $\delta$.
\State \textbf{Initialize}: $\cTpar_1 = \cTint = \varnothing$, $\widehat{\theta}_1^s = \mathbf{0}_d$, $\widehat{\theta}_1^b = \mathbf{0}_d$, $\bV_1 = \mathbf{I}_d$, $K = \lceil 2P/\epsilon\rceil + 3$, $\cK = \llbracket -K,K\rrbracket$, $\cTSE_{s,k}=\cTSE_{b,k} = \varnothing$ for all $k \in \cK$, $N^{s, k}_1 = N^{b, k}_1 = 0$ for all $k\in\cK$, $\widehat{F}_1^k = \widehat{D}_1^k = 0$ for all $k\in\cK$.

\vspace{.5em}
\While{$t \leq T$}
\If{$\Vert x_t\Vert_{\bV_t^{-1}} > \mu$} \Comment{Estimate the parameters $\theta^b$ and $\theta^s$}
    \State Draw $p_t \sim \cU([-P,P])$ and post $(p_t,p_t)$
    \State Update $\cTpar_{t+1} = \cTpar_t \cup \{t\}$
    \State Update $\bV_{t+1} =  \left(\sum_{l \in \cTpar_{t+1}} x_l x_l^{\top} +  \bI_d \right)^{-1}$
    \State Update parameter estimates
    \[
    \textstyle{
    \widehat{\theta}^s_{t+1} = 2P\bV_{t+1}\sum_{l \in \cTpar_{t+1}}\left(\mathbb{I}\{p_l\leq s_l\} - \frac{1}{2}\right)x_l,\quad
    \widehat{\theta}^b_{t+1} = 2P\bV_{t+1}\sum_{l \in \cTpar_{t+1}}\left(\mathbb{I}\{p_l \leq b_l\}-\frac{1}{2}\right)x_l
    }
    \]
\ElsIf{$\vert \cTint \vert < \Tint$} \Comment{Estimate the integrals $I$ and $J$}
\State Draw $p_t \sim \cU([-P,P])$ and post $(p_t,p_t)$
\State Update $\cTint = \cTint \cup \{ t \}$
\If{$\vert \cTint \vert = \Tint$}
\For{$k\in \cK$}
\[
\textstyle{
\widehat{I}^k = \frac{2P}{\Tint}\sum_{l \in \cTint}\mathbb{I}\left\{k\epsilon + x_l^{\top}\widehat{\theta}^b \leq p_l \leq b_l  \right\},\quad
\widehat{J}^k = \frac{2P}{\Tint}\sum_{l \in \cTint}\mathbb{I}\left\{s_l \leq p_l \leq k\epsilon + x_l^{\top}\widehat{\theta}^s \right\}
}
\]
\EndFor
\EndIf
\Else \Comment{Run Successive Elimination}
\State Set $\cA_t = \mleft\{(k, k') \in \cK^2 : k' = \left\lfloor \nicefrac{(x_t^{\top} \left(\widehat{\theta}^s_t - \widehat{\theta}^b_t\right) + k\epsilon)}{\epsilon} \right\rfloor\mright\}$
\For{$(k, k') \in \cA_t$}
\[\UCB_t(k,k') = \widehat{I}^{k'}\widehat{F}_t^k+\widehat{J}^{k} \widehat{D}_t^{k'} + (12PL + 7)\epsilon + 2P\left(\sqrt{\nicefrac{2\log(1/\delta)}{N^{s,k}_t}} + \sqrt{\nicefrac{2\log(1/\delta)}{N^{b,k'}_t}}\right),\]
\[\LCB_t(k,k') = \widehat{I}^{k'}\widehat{F}_t^k+\widehat{J}^{k} \widehat{D}_t^{k'} - (12PL + 7)\epsilon - 2P\left(\sqrt{\nicefrac{2\log(1/\delta)}{N^{s,k}_t}} + \sqrt{\nicefrac{2\log(1/\delta)}{N^{b,k'}_t}}\right)\]
\EndFor
\State Set $\cK_t =\{(k, k')\in \cA_t : \UCB_t(k,k') \geq \max \{\LCB_t(l,l')\ : \ (l, l') \in \cA_t\}\}$
\State Set $k_t^s = \argmin \left\{N_t^{s,k} : (k, k') \in \cK_t \right\}$, and $k_t^b = \argmin \left\{N_t^{b,k'} : (k, k') \in \cK_t \right\}$
\If{$N_t^{s,k_t^s} \leq  N_t^{b,k_t^b}$}
\State $k_t = k_t^s$, and set $k_t'$ to be such that $(k_t, k_t') \in \cK_t$
\State $\cTSE_{s,k_t}  = \cTSE_{s,k_t} \cup \{t\}$
\Else
\State $k'_t = k_t^b$, and set $k_t$ to be such that $(k_t, k_t') \in \cK_t$
\State $\cTSE_{b,k'_t}  = \cTSE_{b,k'_t} \cup \{t\}$
\EndIf
\State Set $p_t = x_t^{\top}\widehat{\theta}^s_t + k_t\epsilon$ and post $(p_t,p_t)$
\State Update $\widehat{F}_{t+1}^{k_t} = \frac{N^{k_t, s}_t\widehat{F}_{t}^{k_t}+\mathbb{I}\left\{s_t \leq p_t\right\}}{N^{k_t, s}_t + 1}$, $\widehat{D}_{t+1}^{k'_t} = \frac{N^{k'_t, b}_t\widehat{D}_{t}^{k'_t}+\mathbb{I}\left\{b_t \geq p_t\right\}}{N^{k'_t, b}_t + 1}$
\State Update $N^{k_t, s}_{t+1} = N^{k_t, s}_t + 1$, and $N^{k'_t, b}_{t+1} = N^{k'_t, b}_t + 1$
\EndIf
\State Quantities that have not been updated during round $t$ are kept the same for round $t+1$
\EndWhile
\end{algorithmic}
\end{algorithm}

Before analyzing the \textsc{SBIP} algorithm, we introduce the different quantities appearing in the algorithm. 
In words, $\cTpar_t$ is the set of indices of the previous rounds that have been spent estimating the parameters $\theta^s$ and $\theta^b$ at the beginning of round $t$. Similarly, $\widehat{\theta}^s_t$, $\widehat{\theta}^b_t$, and $\bV_t$ are the estimates of the parameters and the corresponding empirical covariance matrix at the beginning of round $t$. The set $\cTint$ consists of the indices of the rounds used to estimate the integrals $I$ and $J$. Moreover, $\widehat{I}^k$ estimates $I(k\epsilon)$, and $\widehat{J}^k$ estimates $J(k\epsilon)$. The quantities $\widehat{F}^k_t$ and $\widehat{D}^k_t$ represent our estimates of $F^s(k\epsilon)$ and $D^b(k\epsilon)$ at the beginning of round $t$.
Similarly, $N^{s, k}_t$ (resp., $N^{b, k}_t$) counts the number of rounds in the successive elimination phase where the increment $p_l - x_l^{\top}\widehat{\theta}^s_l$ was equal to $k\epsilon$ (resp., where the increment $p_l - x_l^{\top}\widehat{\theta}^b_l$ was close to $k\epsilon$), and where we gained information on $F^s(k\epsilon)$ (resp., on $D^b(k\epsilon)$). The set $\cA_t$ collects the pairs $(k,k') \in \cK$ such that $x_t^{\top}\widehat{\theta}^s_t + k\epsilon \approx x_t^{\top}\widehat{\theta}^b_t + k'\epsilon$ is a possible price (within $[-P,P]$). Then, for $(k,k') \in \cA_t$, the quantities $\UCB_t(k, k')$ and $\LCB_t(k, k')$ provide upper and lower confidence bounds on the expected gain from trade corresponding to this price.

The third phase is a \emph{successive elimination phase}: at each round, we consider a set of possible optimal prices, with corresponding upper confidence bound larger than the highest lower confidence bound. This set is denoted by $\cK_t$.
In order to ensure sufficient exploration of all potentially optimal price increments $(k\epsilon, k'\epsilon)$ for $(k,k') \in \cK_t$, we choose the price increment with the widest confidence interval: this corresponds to choosing the pair $(k,k') \in \cK_t$ such that either $N_t^{s,k}$ or $N_t^{b,k'}$ is the lowest of the set. We denote $(k_t, k'_t)$ the pair of increments chosen in such way. In order to analyze this phase, we store the indexes of the rounds where we chose this pair because $N_t^{s,k_t}$ was the smallest in the set $\cTSE_{s,k_t}$. Analogously, the rounds where we chose this pair because $N_t^{b,k'_t}$ was the smallest are stored in the set $\cTSE_{b,k'_t}$).

\subsection{Regret Analysis}

The beginning of the proof of Theorem \ref{thm:noise_2_bits} is similar to that of Theorem \ref{thm:etc}. As in the previous case, we define the event
\begin{eqnarray*}
\cE \defeq \Big \{ \forall k \in \cK, \ \forall t \notin \cTpar_{T+1}\cup \cTint, &\left\vert x_t^{\top}(\widehat{\theta}^b_t- \theta^b) \right\vert \leq \epsilon,\,\,  \left\vert x_t^{\top}(\widehat{\theta}^s_t- \theta^s) \right\vert \leq \epsilon,\\
&\left\vert\widehat{I}^k - I\left(k\epsilon\right) \right\vert \leq 2\epsilon, \,\, \left\vert\widehat{J}^k - J(k\epsilon) \right\vert \leq 2\epsilon\Big\}.
\end{eqnarray*} 
Moreover, we define a new event $\cE^{\textnormal{SBIP}}$  as
\begin{eqnarray*}
\cE^{\textnormal{SBIP}} \defeq \cE \bigcap \Big \{\forall k \in \cK, \ \forall t \notin \cTpar_{T+1}\cup \cTint, &\left\vert \widehat{F}^{k}_t - F^s(k\epsilon) \right\vert \leq \sqrt{\frac{2\log(1/\delta)}{N^{s,k}_t}} + L\epsilon, \\
&\left\vert \widehat{D}^{k} - D^s(k\epsilon) \right\vert \leq \sqrt{\frac{2\log(1/\delta)}{N^{b,k}_t}} + 2L\epsilon\Big\}.
\end{eqnarray*} 

The following Lemma shows that the event $\cE^{\textnormal{SBIP}}$ happens with large probability. The detailed proof can be found in \Cref{app:proof event sbip prob}.

\begin{restatable}{lemma}{eventsbipprob}\label{lem:bound_epsilon}
For the choice $\mu = \epsilon\left(P\mleft(d\log\left(\frac{1 + B^2T}{\delta}\right)\mright)^{\nicefrac12} + A\right)^{-1}$ and $\Tint = 8P^2\log(1/\delta)/\epsilon^2$, it holds that
$$\mathbb{P}\left(\cE^{\textnormal{SBIP}}\right)\geq 1-2\delta - 4\delta(2K+1) -   4\delta(2K+1)T.$$
\end{restatable}
Note that the choice  $\delta = \left(T(38 +16P\epsilon^{-1}+ 16P\epsilon^{-1}T + 36T)\right)^{-1}$ ensures that the event $\cE^{\textnormal{SBIP}}$ happens with probability at least $1-T^{-1}$.

The following Lemma (whose proof is located in \Cref{app:proof ucb lcb}) shows that the upper and lower confidence bounds used in Algorithm \ref{alg:SBLBT_detailed} hold conditioned on the high-probability event $\cE^{\textnormal{SBIP}}$.

\begin{restatable}{lemma}{lcbucb}\label{lem:UCB_LCB}
Under the assumptions of Lemma \ref{lem:bound_epsilon} and conditioned on the event $\cE^{\textnormal{SBIP}}$, we have that for all $t \notin \cTpar_{T+1}\cup \cTint$, and all $(k, k')\in \cA_t$;
\begin{eqnarray*}
    \LCB_t(k, k') \leq \egft(x_t, x_t^{\top}\widehat{\theta}_t^s + k\epsilon) \leq \UCB_t(k, k').
\end{eqnarray*}
Moreover, it holds that $(k^*_t, k'^*_t)\in \cK_t$, where we recall that
$$(k^*_t, k'^*_t) \in \argmax_{(k, k') \in \cA_t} \egft(x_t, x_t^{\top}\widehat{\theta}_t^s + k\epsilon).$$
\end{restatable}


Finally, we bound the number of times a sub-optimal price increment $k\epsilon$ can be selected. For $p \in \mathbb{R}$, $x \in \mathbb{R}^d$, recall that we defined 
$$\sr(x, p) = \max_{p'} \egft(x,p') - \egft(x,p).$$

\begin{restatable}{lemma}{lemelim}\label{lem:elim}
Conditioned on the event $\cE^{\textnormal{SBIP}}$, if $t \in \cTSE_{s,k}$, then the following condition holds
\begin{align*}
    \sr(x_t, p_t) \leq (50PL + 28)\epsilon + 16P\sqrt{\frac{2\log(1/\delta)}{N^{s,k}_t}}.
\end{align*}
Similarly, if $t \in \cTSE_{b,k'}$, then the following condition holds
\begin{align*}
    \sr(x_t, p_t) \leq (50PL + 28)\epsilon + 16P\sqrt{\frac{2\log(1/\delta)}{N^{b,k'}_t}}.
\end{align*}
\end{restatable}

The proof for this result can be found in \Cref{app:proof elim}.

We are now ready to bound the regret of \textsc{SBIP}. We begin by decomposing the regret as 
\begin{align*}
    R_T = \sum_{t \in \cTpar_{T+1}} \sr(x_t, p_t) + \sum_{t \in \cTint} \sr(x_t, p_t) + \sum_{t \notin \cTpar_{T+1}\cup \cTint} \sr(x_t, p_t).
\end{align*}
Since the gain from trade is bounded by $2P$, this implies
\begin{align*}
    R_T \leq 2P\vert \cTpar_{T+1}\vert + 2P\vert \cTint\vert + \sum_{t \notin \cTpar_{T+1}\cup \cTint} \sr(x_t, p_t).
\end{align*}
Then, using Lemma \ref{lem:explore_theta}, by setting $\mu = \epsilon\left(P\mleft(d\log\left(\frac{1 + B^2T}{\delta}\right)\mright)^{\nicefrac12} + A\right)^{-1}$ and $\Tint = 8P^2\log(1/\delta)/\epsilon^2$ we have
\begin{multline}\label{eq:explore_phase}
    R_T \leq \epsilon^{-2}\mleft(2Pd \log\left(\frac{T + d}{d}\right)\left(P\sqrt{d\log\left(\frac{1 + B^2T}{\delta}\right)} + A\right)^2\mright) +\\ 16\epsilon^{-2}P^3\log(1/\delta)+ \sum_{t \notin \cTpar_{T+1}\cup \cTint} \sr(x_t,p_t).
\end{multline}
To bound the third term (\ie the regret of the successive elimination phase), we decompose it as
\[
    \sum_{t \notin \cTpar_{T+1}\cup \cTint} \sr(x_t, p_t\epsilon) = \sum_{k \in \cK} \sum_{t \in \cTSE_{s, k}}\sr(x_t, p_t) +  \sum_{k' \in \cK} \sum_{t \in \cTSE_{b, k'}}\sr(x_t, p_t).
\]
For $k \in \cK$, let us partition $\cTSE_{s, k}$ as follows. We denote by $T^{\textnormal{SE}}_{s,k} = \vert \cTSE_{s, k}\vert$, and we define $t_1,\ldots, t_{T^{\textnormal{SE}}_{s,k}}$ the rounds where $t\in \cTSE_{s, k}$. More formally, we have $t_1 < t_2 < ... < t_{T^{\textnormal{SE}}_{s,k}}$, and $\{t_1, ..., t_{T^{\textnormal{SE}}_{s,k}}\} =\cTSE_{s, k}$. We define $\overline{a} = \left\lfloor -\log_2\left(2(50PL+28)\epsilon\right)\right\rfloor$, and for $a \in \llbracket 1, \overline{a}\rrbracket$, we define
$$\mathfrak{t}_a = 2\log(1/\delta)\cdot\left(32P2^a\right)^2.$$
With these notation, we have
\begin{align*}
\sum_{t \in \cTSE_{s, k}}\sr(x_t, p_t) &= \sum_{i \leq \mathfrak{t}_{1}\land T^{\textnormal{SE}}_{s,k}} \sr(x_{t_i}, p_{t_i})\\
&+ \sum_{a = 1}^{\overline{a}-1}\,\,\sum_{\mathfrak{t}_a\land T^{\textnormal{SE}}_{s,k} < i \leq \mathfrak{t}_{a+1}\land T^{\textnormal{SE}}_{s,k}}\sr(x_{t_i}, p_{t_i}) + \sum_{(\mathfrak{t}_{\overline{a}} + 1)\land T^{\textnormal{SE}}_{s,k}\leq i \leq T^{\textnormal{SE}}_{s,k}}\sr(x_{t_i}, p_{t_i}).
\end{align*}
For $1\leq a \leq \overline{a}$, if $\mathfrak{t}_a \leq i \leq T^{\textnormal{SE}}_{s,k} $, by definition of $\mathfrak{t}_a$, we have $$(50PL+28)\epsilon\leq 16P\sqrt{\frac{2\log(1/\delta)}{N^{s,k}_{t_i}}}\leq \frac{2^{-a}}{2},$$ and Lemma \ref{lem:elim} implies that, conditioned on the high-probability event $\cE^{\textnormal{SBIP}}$,
$$\sr(x_{t_i}, p_{t_i}) \leq 2^{-a}.$$
Using the fact that $\sr(x_t, p_t)\leq 2P$, we obtain
\begin{align*}
\sum_{t \in \cTSE_{s, k}} \sr(x_t, p_t) &\leq 2P \mathfrak{t}_{1}+ \sum_{a = 1}^{\overline{a}-1}2^{-a} \left( \mathfrak{t}_{a+1} - \lceil \mathfrak{t}_a\rceil + 1 \right) +  2^{-\overline{a}}\cdot T^{\textnormal{SE}}_{s,k}\\
&\leq 2P \mathfrak{t}_{1}+ \sum_{a = 1}^{\overline{a}-1}2^{-a} \left( \mathfrak{t}_{a+1} - \mathfrak{t}_a + 1 \right) +  2^{-\overline{a}}\cdot T^{\textnormal{SE}}_{s,k}.
\end{align*}
Then, 
\begin{align*}
    \sum_{a = 1}^{\overline{a}-1}2^{-a} \left(\mathfrak{t}_{a+1} - \mathfrak{t}_a\right) &= 6\log(1/\delta)\cdot\left(32P\right)^2\sum_{a = 1}^{\overline{a}-1}2^{a}\\
    &= 6\log(1/\delta)\cdot\left(32P\right)^2 2^{\overline{a}}\\
    &\leq \frac{3 \log(1/\delta)\cdot\left(32P\right)^2}{(50PL+28)\epsilon}
\end{align*}
Moreover, we have that
\begin{align*}
    2^{-\overline{a}}\cdot T^{\textnormal{SE}}_{s,k} \leq 2(50PL+28)\epsilon T^{\textnormal{SE}}_{s,k}.
\end{align*}
Therefore,
\begin{align*}
\sum_{t \in \cTSE_{s, k}} \sr(x_t, p_t) \leq 4P\log(1/\delta)\cdot\left(128P\right)^2 +  \frac{3 \log(1/\delta)\cdot\left(32P\right)^2}{(50PL+28)\epsilon} + 1 + 2(50PL+28)\epsilon T^{\textnormal{SE}}_{s,k}.
\end{align*}
We proceed similarly to bound $\sum_{t \in \cTSE_{b, k}} \sr(x_t, p_t)$, and we obtain that
\begin{align*}
\sum_{t \in \cTSE_{b, k}} \sr(x_t, p_t) \leq 4P\log(1/\delta)\cdot\left(128P\right)^2 +  \frac{3 \log(1/\delta)\cdot\left(32P\right)^2}{(50PL+28)\epsilon} + 1 + 2(50PL+28)\epsilon T^{\textnormal{SE}}_{b,k}.
\end{align*}
Summing over $k \in \cK$ and $k' \in \cK$, and using the fact that $\sum_{k \in \cK}T^{\textnormal{SE}}_{s,k} + T^{\textnormal{SE}}_{b,k}\leq T$, we find
\begin{align}\label{eq:elim_phase}
    \sum_{t \notin \cTpar_{T+1}\cup \cTint} \sr(x_t, p_t\epsilon) &\leq 2\vert \cK\vert \left(4P\log(1/\delta)\times\left(128P\right)^2 +  \frac{3 \log(1/\delta)\times\left(32P\right)^2}{(50PL+28)\epsilon} + 1 \right)\nonumber\\
    &+ 2(50PL+28)\epsilon T.
\end{align}
Combining Equations \eqref{eq:explore_phase} and \eqref{eq:elim_phase}, we find that, conditioned on the event $\cE^{\textnormal{SBIP}}$,
\begin{align*}
    R_T \leq& \frac{2Pd \log\left(\frac{T + d}{d}\right)\left(P\sqrt{d\log\left(\frac{1 + B^2T}{\delta}\right)} + A\right)^2}{\epsilon^2} + \frac{18P^3\log(1/\delta)}{\epsilon^2}\\
    &+2\vert \cK\vert \left(4P\log(1/\delta)\times\left(128P\right)^2 +  \frac{3 \log(1/\delta)\times\left(32P\right)^2}{(50PL+28)\epsilon} + 1 \right)+ 2(50PL+28)\epsilon T.
\end{align*}

By setting $\epsilon = (d^2\log(T)^2/T)^{\frac{1}{3}}$ and $\delta = T^{-3}$, and using the definition of $\cK$, we find that there exists a constant $\widetilde C'$ depending on $A$, $B$ $C$, and $L$ such that with probability $1 - 20T^{-2/3}$, 
\begin{align*}
    R_T \leq \widetilde C'(d\log(T)T)^{\nicefrac{2}{3}}.
\end{align*}
This concludes the proof of \Cref{thm:noise_2_bits}.

\section{Proofs of auxiliary lemmas}\label{app:aux}
\subsection{Proof of Lemma \ref{lem:EGFT}}\label{aux:egft dec}

\egftdec*

\begin{proof}
The proof relies on the characterization of the gain from trade given by Lemma 4.1 in \cite{CesaBianchiCCF21}. Recall that, under Assumption \ref{ass:bounded_context}, we have $\xi_t^s \in [-C,C]$, and $\xi_t^b \in [-C,C]$. Then, denoting $f^s$ and $f^b$ the densities of $\xi_t^s$ and $\xi_t^b$, we have
\begin{small}
\begin{align*}
&\egft(x,p)\\
&\hspace{0.5cm}= \int_{(s,b) \in [x^{\top}\theta^s-C,x^{\top}\theta^s+C]\times [x^{\top}\theta^b-C,x^{\top}\theta^b+C]}(b - s)\mathbb{I}\{s \leq p \leq b\}f^b(x^{\top}\theta^b-s) f^s(x^{\top}\theta^s-s) \diff s\diff b.
\end{align*}    
\end{small}
Let $\tilde{f}^b : b \mapsto f^b(x^{\top}\theta^b-s)$ and $\tilde{f}^s : s \mapsto f^s(x^{\top}\theta^s-s)$ be the densities of $s_t$ and $b_t$ conditionnally on $x_t = x$, and note that outside of $[x^{\top}\theta^b-C,x^{\top}\theta^b+C]\times [x^{\top}\theta^s-C,x^{\top}\theta^s+C]$, $\tilde{f}^s(s)\tilde{f}^b(b) = 0$. With these notation, we have
\begin{align*}
\egft(x,p) &= \int_{(s,b) \in [-P,P]^2}(b - s)\mathbb{I}\{s \leq p \leq b\}\tilde{f}^b(b) \tilde{f}^s(s) \diff s \diff b\\
&= \int_{(s, b)\in [-P,p]\times [p, P]} \left(\int_{-P}^b du  - \int_{-P}^s du\right)\tilde{f}^s(s)\tilde{f}^b(b) \diff s \diff b.
\end{align*}
Thus,
\begin{align*}
\egft(x,p) =& \int_{[-P,P]}\left(\int_{(s,b) \in[-P,p]\times [p\lor u, P]} \tilde{f}^s(s) \tilde{f}^b(b) \diff s \diff b\right)  \diff u  \\
&-  \int_{[-P,p]}\left(\int_{(s,b) \in[u,p]\times [p, P]} \tilde{f}^s(s) \tilde{f}^b(b) \diff s \diff b \right) \diff u.
\end{align*}

This yields
\begin{align*}
\egft(x,p) =& \int_{[-P,p]}\left(\int_{(s,b) \in[-P,p]\times [p, P]} \tilde{f}^s(s) \tilde{f}^b(b)\diff s \diff b  \right) \diff u\\
&+ \int_{[p, P]}\left(\int_{(s,b) \in[-P,p]\times [u, P]} \tilde{f}^s(s) \tilde{f}^b(b) \diff s \diff b  \right) \diff u  \\
& -  \int_{[-P,p]}\left(\int_{(s,b) \in[u,p]\times [p, P]} \tilde{f}^s(s) \tilde{f}^b(b) \diff s \diff b  \right) \diff u.
\end{align*}
This, in turn, implies that
\begin{align*}
\egft(x,p) =& \int_{[-P,p]}\left(\int_{(s,b) \in[-P,u]\times [p, P]} \tilde{f}^s(s) \tilde{f}^b(b)\diff s \diff b  \right) \diff u \\
&+ \int_{[p, P]}\left(\int_{(s,b) \in[-P,p]\times [u, P]} \tilde{f}^s(s) \tilde{f}^b(b) \diff s \diff b  \right) \diff u\\
=&
 \int_{[p,P]}\tilde{f}^b(b) \diff d \int_{[-P,p]}\int_{[-P,u]} \tilde{f}^s(s) \diff s  \diff u \\&+ \int_{[-P, p]}\tilde{f}^s(s)\diff s \int_{[p,P]}\int_{[u, P]}  \tilde{f}^b(b)  \diff b  \diff u .
\end{align*}
By using the change in variables $\xi^s = s- x^{\top}\theta^s$, $u^s = u- x^{\top}\theta^s$, $\xi^b = b- x^{\top}\theta^b$, and $u^b = u- x^{\top}\theta^b$, and the definition of $\tilde{f}^s$ and $\tilde{f}^b$, we get
\begin{align*}
\egft(x,p) =&  \int_{[p- x^{\top}\theta^b,P- x^{\top}\theta^b]}f(\xi^b) \diff d \int_{[-P, p]}\int_{[-P- x^{\top}\theta^s,u - x^{\top}\theta^s]} f(\xi^s)\diff \xi^s  \diff u\\
&+\int_{[-P- x^{\top}\theta^s, p- x^{\top}\theta^s]}f(\xi^s)\diff \xi^s \int_{[p,P]}\int_{[u- x^{\top}\theta^b, P- x^{\top}\theta^b]}  f(\xi^b)  \diff \xi^b  \diff u \\
=&  \int_{[p- x^{\top}\theta^b,P- x^{\top}\theta^b]}f(\xi^b) \diff \xi^b \int_{[-P- x^{\top}\theta^s, p- x^{\top}\theta^s]}\int_{[-P- x^{\top}\theta^s,u^s]} f(\xi^s)\diff \xi^s  \diff u^s\\
&+\int_{[-P- x^{\top}\theta^s, p- x^{\top}\theta^s]}f(\xi^s)\diff \xi^s \int_{[p- x^{\top}\theta^b,P- x^{\top}\theta^b]}\int_{[u^b, P- x^{\top}\theta^b]}  f(\xi^b)  \diff \xi^b  \diff u^b.
\end{align*}
Now, under Assumption \ref{ass:bounded_context}, $f^b$ and $f^s$ are null outside of $[-C,C]$. Moreover, $P- x^{\top}\theta^b \geq C$ and $-P- x^{\top}\theta^s \leq -C$. Thus,
\begin{align*}
\egft(x,p) =& \int_{[p- x^{\top}\theta^b,C]}f(\xi^b) \diff d \int_{[-C, p- x^{\top}\theta^s]}\int_{[-C,u^s]} f(\xi^s)\diff \xi^s  \diff u^s\\
&+ \int_{[-C, p-x^{\top}\theta^s]}f(\xi^s)\diff \xi^s \int_{[p- x^{\top}\theta^b,C]}\int_{[u^b, C]}  f(\xi^b)  \diff \xi^b  \diff u^b \\
=&  D^b(p- x^{\top}\theta^b)\int_{-C}^{p- x^{\top}\theta^s}F^s(u^s)  \diff u^s + F^s(p-x^{\top}\theta^s) \int_{p- x^{\top}\theta^b}^CD^b(u^b)  \diff u^b.
\end{align*}
Using the definition of $\delta^s = p- x^{\top}\theta^s$ and $\delta^b= p- x^{\top}\theta^b$, we obtain the desired result.
\end{proof}

\subsection{Proof of Lemma \ref{lem:explore_theta}}\label{app:proof lemma stop tpar}

\stopTpar*
\begin{proof}
The elliptical potential Lemma (see, \eg Proposition 1 in \cite{carpentier2020elliptical}) implies that almost surely,
\begin{align*}
    \sum_{t\in \cTpar_{T+1}} \Vert x_t\Vert_{\left(\sum_{l \in \cTpar_t} x_l x_l^{\top} +  \bI_d \right)^{-1}} \leq \sqrt{\vert\cTpar_{T+1}\vert d \log\left(\frac{\vert\cTpar_{T+1}\vert + d}{d}\right)}.
\end{align*}
Since for all $t \in \cTpar_{T+1}$, $\Vert x_t\Vert_{\left(\sum_{l \in \cTpar_t} x_l x_l^{\top} +  \bI_d \right)^{-1}} \geq \mu$, this implies that
\begin{align*}
    \vert \cTpar_{T+1}\vert \mu \leq \sqrt{\vert\cTpar_{T+1}\vert d \log\left(\frac{\vert\cTpar_{T+1}\vert + d}{d}\right)}.
\end{align*}
Now, almost surely, $\vert \cTpar_{T+1}\vert \leq T$. This implies in particular that
\begin{align*}
    \vert \cTpar_{T+1}\vert \leq \frac{d \log\left(\frac{T + d}{d}\right)}{\mu^2},
\end{align*}
which concludes the proof.
\end{proof}
\subsection{Proof of Lemma \ref{lem:bound_epsilon_etc}}\label{app:bound event etc lemma}
\boundetcevent*
\begin{proof}
The Lemma is a consequence of the following two auxiliary results, bounding respectively the error for estimating the parameters $\theta^s$ and $\theta^b$, and the integrals $I$ and $J$.
\begin{restatable}{lemma}{eventE}\label{lem:eventE}
Let $\cE_1$ be the event :
\begin{align*}
   \cE_1 \defeq \mleft\{\forall t \notin \cTpar_{T+1},\quad \left\vert x_t^{\top}(\widehat{\theta}^b_t- \theta^b) \right\vert \leq \epsilon \quad \textnormal{ and } \quad \left\vert x_t^{\top}(\widehat{\theta}^s_t- \theta^s) \right\vert \leq \epsilon \mright\}.
\end{align*}
Then, for the choice $\mu = \epsilon\left(P\mleft(d\log\left(\frac{1 + B^2T}{\delta}\right)\mright)^{\nicefrac12} + A\right)^{-1}$, we have $\mathbb{P}\left(\cE_1\right)\geq 1-2\delta$. 
\end{restatable}
The proof of the above lemma is deferred to \Cref{app:proof lem event E}.

\begin{restatable}{lemma}{eventF}\label{lem:eventF}
Let $\cE_2$ be the event:
\begin{align*}
 \cE_2 \defeq \mleft\{\vert \cTint\vert < \Tint\textnormal{ or }\forall k \in \cK\,\, \ \left\vert\widehat{I}^k - I\left(k\epsilon\right) \right\vert \leq 2\epsilon  \textnormal{ and } \left\vert\widehat{J}^k - J(k\epsilon) \right\vert \leq 2\epsilon \mright\} \cap \cE_1.
\end{align*}
Then, for the choice $\mu = \epsilon\left(P\mleft(d\log\left(\frac{1 + B^2T}{\delta}\right)\mright)^{\nicefrac12} + A\right)^{-1}$ and $\Tint = 8P^2\log(1/\delta)/\epsilon^2$, we have 
\[\mathbb{P}\left(\cE_2\right)\geq 1-2\delta - 4\delta(2K+1).\]
\end{restatable}

The proof of the above lemma is deferred to \Cref{app:proof lemma eventF}.

Next, we control the error $\vert \widehat{F}^k - F^s(k\epsilon)\vert$ uniformly for $k \in \cK$. The results for $\vert \widehat{D}^k -D^b(k\epsilon)\vert$ follow from similar arguments. To do so, we rely on the following well known result (for the sake of completeness, we provide a proof in \Cref{app:proof hoeffding rigorous}).

\begin{restatable}{lemma}{Hoeffding}\label{lem:Hoeffding_rigorous}
Let $(y_t)_{t\geq 1}$ be a sequence of random variables adapted for a filtration $\cF_t$, such that $y_t - \mathbb{E}\left[y_t \vert \cF_{t-1}\right] \in [m, M]$. Assume that for $t \in \mathbb{N}_*$, $\iota_t \in \{0,1\}$ is $\cF_{t-1}$-measurable, and define $N_t \defeq \sum_{l \leq t}\iota_l$, and $\widehat{\mu}_t \defeq N_t^{-1}\sum_{l \leq t}\iota_l (y_l - \mathbb{E}\left[y_l \vert \cF_{l-1}\right] )$ if $N_t \geq 1$. Then, for any $t \in \mathbb{N}_*$ and $\delta \in (0,1)$,
$$\mathbb{P}\left(N_t = 0 \,\,\textnormal{ or }\,\, \vert \widehat{\mu}_t\vert\leq (M-m)\sqrt{\frac{\log(1/\delta)}{2N_t }} \right) \geq 1-2t\delta.$$
Moreover, for any $t >0$ and $\delta \in (0,1)$,
$$\mathbb{P}\left(N_t = t \,\,\textnormal{ and }\,\, \vert \widehat{\mu}_t\vert\geq (M-m)\sqrt{\frac{\log(1/\delta)}{2N_t }} \right) \leq 1-2\delta.$$
\end{restatable}

For $t\leq T$, we define $\iota_t \defeq \mathbb{I}\left\{t \in \cT^F_k\right\}$, $y_t \defeq \mathbb{I}\left\{s_t \leq p_t \right\}$, and we observe that for $\cF_t = \sigma\left((x_1,\ldots, x_{t+1}, (\mathbb{I}\{s_1 \leq p_1\},\mathbb{I}\{b_1 \leq p_1\}),\ldots, (\mathbb{I}\{s_t \leq p_t\},\mathbb{I}\{b_t \leq p_t\}) \right)$, $\iota_t$ is $\cF_{t-1}$-measurable, and $y_t$ is $\cF_t$ adapted. Moreover,
\begin{align*}
    \iota_t\mathbb{E}\left[y_t \,\,\vert\,\, \cF_{t-1}\right] &= \iota_t\mathbb{P}\left(x_t^{\top} \theta^s + \xi_t^s \leq x_t^{\top} \widehat{\theta}^s_t + k\epsilon \right)\\
    & = \iota_tF^s\left(x_t^{\top} \left(\widehat{\theta}^s_t - \theta^s\right) + k\epsilon \right).
\end{align*}
Using Lemma \ref{lem:Hoeffding_rigorous}, we find that with probability $1 - 2\delta$, either $\vert \cTF_k\vert < \TFD$, or
\begin{align*}
    \left\vert \widehat{F}^k_{t} -\frac{\sum_{s\leq t}\iota_tF^s\left(x_t^{\top} \left(\widehat{\theta}^s_t - \theta^s\right) + k\epsilon \right)}{\TFD} \right\vert \leq \sqrt{\frac{2\log(1/\delta)}{\TFD}}.
\end{align*}
Moreover, on the event $\cE_1$, for all $t \notin \cTpar$, $\vert x_t^{\top} (\widehat{\theta}^s_t - \theta^s)\vert \leq \epsilon$. Using the fact that $F^s$ is $L$-Lipschitz, we find that
\begin{align*}
    \left\vert F^s\left(x_t^{\top} \left(\widehat{\theta}^s_t - \theta^s\right) + k\epsilon \right) - F^s\left(k\epsilon \right)\right\vert \leq L\epsilon.
\end{align*}
Thus, with probability $1 - 2\delta$, either $\vert \cTF_k\vert < \TFD$, or
\begin{align*}
    \left\vert \widehat{F}^k_{t} - F^s(k\epsilon)\right\vert \leq \sqrt{\frac{2\log(1/\delta)}{\TFD}} + L\epsilon.
\end{align*}
Using Lemma \ref{lem:eventE} and \ref{lem:eventF}, along with a union bound over $k \in \cK$ yields the desired result for the choice $\TFD = 2\log(1/\delta)\epsilon^{-2}$.
\end{proof}
\subsection{Proof of Lemma \ref{lem:conf_bound_etc}}\label{app:proof conf bound etc}

\errgftevent*

\begin{proof}
The proof relies on the following result :
\begin{restatable}{lemma}{decompoerror}\label{lem:decompo_error}
On the event $\cE^{\textnormal{ETC}}$, for all $(k,k')\in \cA_t$, and $p = x_t^{\top}\widehat{\theta}^s_t + k\epsilon$, we have 
\begin{align*}
\left \vert \egft(x_t, p)  - \left(\widehat{I}^{k'}\widehat{F}^k + \widehat{J}^k\widehat{D}^{k'}\right)\right \vert& \leq 6PL\epsilon +3\epsilon + 2P\left(\vert\widehat{\Delta} F\vert+\vert\widehat{\Delta} D\vert\right) +  \vert\widehat{\Delta} I\vert + \vert\widehat{\Delta}J\vert
\end{align*}
where we define $\widehat{\Delta}I = I(k' \epsilon) -\widehat{I}^{k'}$, $\widehat{\Delta}J = J(k \epsilon) -\widehat{J}^{k}$, $\widehat{\Delta}F = F(k \epsilon) -\widehat{F}^{k}$, $\widehat{\Delta}D = D(k' \epsilon) -\widehat{D}^{k'}$.
\end{restatable}

The proof of this intermediate result is in \Cref{app: proof decompo error}.

To conclude the proof of Lemma \ref{lem:conf_bound_etc}, it remains to bound the gaps $\widehat{\Delta} F$, $\widehat{\Delta} D$, $\widehat{\Delta} I$, and $\widehat{\Delta} J$ on the event $\cE^{\textnormal{ETC}}$. By definition of the event $\cE^{\textnormal{ETC}}$, on this event $\widehat{\Delta}I \leq 2 \epsilon$, $\widehat{\Delta}D  \leq (L+1)\epsilon$, $\widehat{\Delta}J \leq 2\epsilon$ and $\widehat{\Delta}F \leq (L+1)\epsilon$. Thus, on the event $\cE^{\textnormal{ETC}}$,
\begin{align*}
\left \vert \egft(x_t, p)  - \left(\widehat{I}^{k'}\widehat{F}^k + \widehat{J}^k\widehat{D}^{k'}\right)\right \vert& \leq 6PL\epsilon +3\epsilon + 4P(L+1)\epsilon +4 \epsilon.
\end{align*}
This concludes the proof.
\end{proof}

\subsection{Proof of Lemma \ref{lem:discretisation_error}}\label{app:proof discretization error}

\errdiscrevent*

\begin{proof}
By definition of $\cA_t$,
$$\sup_{p\in [-P, P]}\egft(x_t, p) \geq  \egft(x_t, x_t^{\top}\widehat{\theta}^s_t + k_t^*\epsilon).$$
Since the mapping $p \mapsto \egft(x, p)$ is continuous, we can define
$$p^*_t \in \argmax_{p}\egft(x_t, p), \quad \tilde{k}_t =\left\lfloor\frac{p^*_t - x_t^{\top}\widehat{\theta}^s_t}{\epsilon}\right\rfloor, \quad\text{and }\,\, \tilde{k'}_t = \left\lfloor \frac{x_t^{\top}\left(\widehat{\theta}^s_t - \widehat{\theta}^b_t\right)+\tilde{k}_t}{\epsilon} \right\rfloor.$$
We now show that $(\tilde{k}_t, \tilde{k}'_t)\in\cA_t$. By Lemma \ref{lem:EGFT}, we have that $p^*_t - x_t^{\top}\theta^s \geq -C$, and  $p^*_t - x_t^{\top}\theta^b \leq C$ (otherwise $\egft(x_t, p^*_t) = 0$). This implies that $p^*_t - x_t^{\top}\theta^s \in [-2P, 2P]$, and similarly that $p^*_t - x_t^{\top}\theta^b \in [-2P, 2P]$.

On the one hand, on the event $\cE^{\textnormal{ETC}}$, $\left\vert x_t^{\top}\left(\theta^s - \widehat{\theta}^s_t\right) \right\vert \leq \epsilon$, so $p^*_t - x_t^{\top}\widehat{\theta}^s_t \in [-2P-\epsilon,2P+\epsilon]$. Thus, $\tilde{k}_t \in \cK$.
On the other hand,  $$x_t^{\top}\left(\widehat{\theta}^s_t-\widehat{\theta}^b_t\right) + \tilde{k}_t\epsilon = \left(x_t^{\top} \widehat{\theta}^s_t + \tilde{k}_t\epsilon -p_t^*\right) + (p^*_t - x_t^{\top} \theta^b) + x_t^{\top}\left(\theta^b - \widehat{\theta}^b_t\right).$$
Then, on the event $\cE^{\textnormal{ETC}}$, 
\[
x_t^{\top}\left(\widehat{\theta}^s_t-\widehat{\theta}^b_t\right) + \tilde{k}_t\epsilon \in  [-2P-2\epsilon,2P+2\epsilon],\], so $k_t' \in [-2P-3\epsilon,2P+3\epsilon]$. Therefore, $\tilde{k}'_t \in \cK$. This implies that $(\tilde{k}_t,\tilde{k}_t')\in \cA_t$, so
$$\egft(x_t, x_t^{\top}\widehat{\theta}^s_t + k^*_t\epsilon) \geq \egft(x_t, x_t^{\top}\widehat{\theta}^s_t + \tilde{k}_t\epsilon).$$
Finally, Lemma \ref{lem:EGFT} and Assumption \ref{ass:bounded_densities} imply that the function $p \mapsto \egft(x_t, p)$ is $2LP$-Lipschitz continuous. This, in turn, implies that
$$\egft(x_t, p^*_t) - \egft(x_t, x_t^{\top}\widehat{\theta}^s_t + \tilde{k}_t\epsilon) \leq 2LP\epsilon.$$
This proves the statement.
\end{proof}

\subsection{Proof of Lemma \ref{lem:bound_epsilon}}\label{app:proof event sbip prob}

\eventsbipprob*

\begin{proof}
The proof relies on Lemma \ref{lem:eventE} and \ref{lem:eventF}. On top of these results, we have to control the error $\vert \widehat{F}^k_t - F^s(k\epsilon)\vert$ uniformly for $k \in \cK$. To do so, we rely on Lemma \ref{lem:Hoeffding_rigorous}. For $t\leq T$, we define $\iota_t = \mathbb{I}\left\{t \notin \cTint\cup \cTpar_{T+1} \text{ and }k_t = k\right\}$, $y_t = \mathbb{I}\left\{s_t \leq p_t \right\}$, and we note that for $\cF_t = \sigma\left((x_1, \ldots, x_{t+1}, (\mathbb{I}\{s_1 \leq p_1\},\mathbb{I}\{b_1 \leq p_1\}), \ldots, (\mathbb{I}\{s_t \leq p_t\},\mathbb{I}\{b_t \leq p_t\}) \right)$, $\iota_t$ is $\cF_{t-1}$-measurable, and $y_t$ is $\cF_t$ adapted. Moreover,
\begin{align*}
    \iota_t\mathbb{E}\left[y_t \,\vert\, \cF_{t-1}\right] &= \iota_t\mathbb{P}\left(x_t^{\top} \theta^s + \xi_t^s \leq x_t^{\top} \widehat{\theta}^s_t + k\epsilon \right)\\
    & = \iota_tF^s\left(x_t^{\top} \left(\widehat{\theta}^s_t - \theta^s\right) + k\epsilon \right).
\end{align*}
Using Lemma \ref{lem:Hoeffding_rigorous}, we find that with probability $1 - 2\delta t$, either $N^{s,k}_t = 0$ or
\begin{align*}
    \left\vert \widehat{F}^k_{t} -\frac{\sum_{s\leq t}\iota_tF^s\left(x_t^{\top} \left(\widehat{\theta}^s_t - \theta^s\right) + k\epsilon \right)}{N^{s,k}_t} \right\vert \leq \sqrt{\frac{2\log(1/\delta)}{N^{s,k}_t}}
\end{align*}
(where we recall that we adopt the convention $1/0 = \infty$). Moreover, on the event $\cE_1$, for all $t \notin \cTpar$, $\vert x_t^{\top} (\widehat{\theta}^s_t - \theta^s)\vert \leq \epsilon$. Since $F^s$ is $L$-Lipschitz, this implies
\begin{align*}
    \vert F^s\left(x_t^{\top} \left(\widehat{\theta}^s_t - \theta^s\right) + k\epsilon \right) - F^s\left(k\epsilon \right)\vert \leq L\epsilon.
\end{align*}
Thus, with probability $1 - 2\delta t$, 
\begin{align*}
    \left\vert \widehat{F}^k_{t} - F^s(k\epsilon)\right\vert \leq \sqrt{\frac{2\log(1/\delta)}{N^{s,k}_t}} + L\epsilon.
\end{align*}
Using a union bound over $k \in \cK$ and $t \notin \cTint\cup \cTpar_{T+1}$ yields the desired result.

\medskip

Next, we can bound $\vert \widehat{D}^k_{t} - D^b(k\epsilon)\vert$ using similar arguments. In particular, we define $\iota_t = \mathbb{I}\left\{t \notin \cTint\cup \cTpar_{T+1} \text{ and }k'_t = k\right\}$, $y_t = \mathbb{I}\left\{b_t \geq p_t \right\}$ and note that 
\begin{align*}
    \iota_t\mathbb{E}\left[y_t \,\vert\, \cF_{t-1}\right] &= \iota_t\mathbb{P}\left(x_t^{\top} \theta^b + \xi_t^b \geq x_t^{\top} \widehat{\theta}^s_t + k_t\epsilon \ \Big \vert \  \cF_{t-1} \right)\\
    & = \iota_tD^b\left(x_t^{\top} \left(\widehat{\theta}^b_t - \theta^b\right) + k\epsilon + \epsilon\left(\frac{x_t^{\top} \widehat{\theta}^s_t - x_t^{\top}\widehat{\theta}^b_t + k_t\epsilon}{\epsilon} - k\right) \right).
\end{align*}
By definition of $\cA_t$, we have that if $\iota_t = 1$, then $\left \vert \frac{x_t^{\top} \widehat{\theta}^s_t - x_t^{\top}\widehat{\theta}^b_t + k_t\epsilon}{\epsilon} - k\right\vert \leq 1$. Using the Lipschitz continuity of $D^b$, this implies that, conditioned on the event $\cE_1$,
\begin{align*}
    \iota_t\left\vert \mathbb{E}\left[y_t \,\vert \,\cF_{t-1}\right] - D^b(k\epsilon)\right\vert \leq 2L\epsilon.
\end{align*}
The rest of the proof follows similarly.
\end{proof}

\subsection{Proof of Lemma \ref{lem:UCB_LCB}}\label{app:proof ucb lcb}

\lcbucb*
\begin{proof}
Let us prove the first part of Lemma \ref{lem:UCB_LCB}. 
By Lemma \ref{lem:decompo_error}, and by definition of the event $\cE^{\textnormal{SBIP}}$, we have for all $(k, k')\in \cA_t$, and all $t \notin \cTpar_{T+1}\cup \cTint$,
\begin{align*}
\left \vert \egft(x_t, p)  - \left(\widehat{I}^{k'}\widehat{F}^k + \widehat{J}^k\widehat{D}^{k'}\right)\right \vert\\
&\hspace{-1cm} \leq 6PL\epsilon +3\epsilon + 2P\left(\sqrt{\frac{2\log(1/\delta)}{N^{s,k}_t}} 
 + \sqrt{\frac{2\log(1/\delta)}{N^{b,k'}_t}} + 3L\epsilon \right) +  4\epsilon\\
 &\hspace{-1cm}\leq (12PL + 7)\epsilon + 2P\left(\sqrt{\frac{2\log(1/\delta)}{N^{s,k}_t}} 
 + \sqrt{\frac{2\log(1/\delta)}{N^{b,k'}_t}}\right).
\end{align*}

This concludes the first part of Lemma \ref{lem:UCB_LCB}. The second claim follows immediately by noticing that by definition, $(k_t^*, k_t'^*) \in \cA_t$, and that, for all $(k, k') \in \cA_t$,
\begin{eqnarray*}
    \LCB_t(k, k') \leq \egft(x_t, x_t^{\top}\widehat{\theta}_t^s + k\epsilon)\leq \egft(x_t, x_t^{\top}\widehat{\theta}_t^s + k_t^*\epsilon) \leq \UCB_t(k_t^*, k_t'^*).
\end{eqnarray*}
This concludes the proof.
\end{proof}

\subsection{Proof of Lemma \ref{lem:elim}}\label{app:proof elim}

\lemelim*

\begin{proof}
Assume that $t \in \cTSE_{s,k}$. Then, our choice of $k_t$ together with Lemma \ref{lem:UCB_LCB} ensures that, conditioned on the event $\cE^{\textnormal{SBIP}}$, 
$$\UCB_t(k_t,k'_t) \geq \LCB_t(k^*_t, k'^*_t).$$ This implies that
\begin{multline*}
    \LCB_t(k_t,k'_t) + \left(\UCB_t(k_t,k'_t)-\LCB_t(k_t,k'_t)\right) \geq \\\UCB_t(k^*_t, k'^*_t)+ \left(\LCB_t(k^*_t, k'^*_t)-\UCB_t(k^*_t, k'^*_t)\right).
\end{multline*}
Using again Lemma \ref{lem:UCB_LCB}, this implies that
\begin{align*}
    \egft(x_t, x_t^{\top}\widehat{\theta}^s + k_t\epsilon)& + \left(\UCB_t(k_t,k'_t)-\LCB_t(k_t,k'_t)\right) + \left(\UCB_t(k^*_t,k'^*_t)-\LCB_t(k^*_t,k'^*_t)\right)\\
    &\geq  \egft(x_t, x_t^{\top}\widehat{\theta}^s + k^*_t\epsilon).
\end{align*}
Thus,
\begin{align*}
    \egft(x_t, x_t^{\top}\widehat{\theta}^s + k^*_t\epsilon)& - \egft(x_t, x_t^{\top}\widehat{\theta}^s + k_t\epsilon) \\
    &\leq \left(\UCB_t(k_t,k'_t)-\LCB_t(k_t,k'_t)\right) + \left(\UCB_t(k^*_t,k'^*_t)-\LCB_t(k^*_t,k'^*_t)\right).
\end{align*}
Then, conditioned on the event $\cE^{\textnormal{SBIP}}$,
\begin{align*}
    \UCB_t(k_t,k'_t)-\LCB_t(k_t,k'_t)& \leq 2(12PL + 7)\epsilon + 4P\left(\sqrt{\frac{2\log(1/\delta)}{N^{s,k}_t}} + \sqrt{\frac{2\log(1/\delta)}{N^{b,k'}_t}}\right).
\end{align*}
Since $t \in \cTSE_{s,k}$, we know that $N^{b,k'_t}_t\geq N^{s,k_t}_t$. Then, we have that 
\begin{align*}
    \UCB_t(k_t,k'_t)-\LCB_t(k_t,k'_t)& \leq 2(12PL + 7)\epsilon + 8P\sqrt{\frac{2\log(1/\delta)}{N^{s,k}_t}}.
\end{align*}
Similarly, since $t \in \cTSE_{s,k}$, we have that $N^{b,k'^*_t}_t\geq N^{s,k_t}_t$, and $N^{s,k^*_t}_t\geq N^{s,k_t}_t$, so we also have
\begin{align*}
    \UCB_t(k^*_t,k'^*_t)-\LCB_t(k^*_t,k'^*_t) &\leq 2(12PL + 7)\epsilon + 8P\sqrt{\frac{2\log(1/\delta)}{N^{s,k}_t}}.
\end{align*}
Thus, 
\begin{align*}
    \egft(x_t, x_t^{\top}\widehat{\theta}^s + k^*_t\epsilon) - \egft(x_t, x_t^{\top}\widehat{\theta}^s + k_t\epsilon) \leq 4(12PL + 7)\epsilon + 16P\sqrt{\frac{2\log(1/\delta)}{N^{s,k}_t}}.
\end{align*}
By Lemma \ref{lem:discretisation_error}, this implies  
\begin{align*}
    \sr(x_t, x_t^{\top}\widehat{\theta}^s + k_t\epsilon)\leq (50PL + 28)\epsilon + 16P\sqrt{\frac{2\log(1/\delta)}{N^{s,k}_t}}.
\end{align*}
The proof of the second claim follows from similar arguments.
\end{proof}

\subsection{Proof of Lemma \ref{lem:eventE}}\label{app:proof lem event E}

\eventE*

\begin{proof}
    
Let us prove the bound $\vert x_t^{\top}(\widehat{\theta}^s_t- \theta^s) \vert\leq \epsilon$ with high probability; the bound on $\vert x_t^{\top}(\widehat{\theta}^b_t- \theta^b) \vert$ will follow from similar arguments. 

We introduce the variables
$$\tilde{x}_t  \defeq x_t \mathbb{I}\left\{t \in \cTpar_{t+1}\right\} \quad 
\text{and} \quad \tilde{y}_t \defeq 2P\,\mathbb{I}\left\{t \in \cTpar_{t+1}\right\}\left(\mathbb{I}\left\{p_t \leq s_t\right\}-\frac{1}{2}\right)$$
and the $\sigma$-algebra $\cF_t = \sigma\left((x_1,\ldots, x_{t+1}, (\mathbb{I}\{s_1 \leq p_1\},\mathbb{I}\{b_1 \leq p_1\}),\ldots, (\mathbb{I}\{s_t \leq p_t\},\mathbb{I}\{b_t \leq p_t\}) \right).$ With these notation, we have that $\tilde{x_t}$ and $\{t \in \cTpar_{t+1}\}$ are $\cF_{t}$-measurable. Moreover,
\begin{align*}
    \mathbb{E}\left[\tilde{y}_t\vert \cF_{t-1}\right] &= \mathbb{I}\left\{t \in \cTpar_{t+1}\right\}\cdot\left( 2P \int_{-P}^P \mathbb{P}\left[u \leq s_t \vert \cF_{t-1}\right]\frac{\diff u}{2P} - P\right)\\
    &= \mathbb{I}\left\{t \in \cTpar_{t+1}\right\}\cdot \left( \int_{-P}^P\int_{-C}^C \mathbb{I}\left\{u \leq  x_t^{\top}\theta^s + \xi\right\} f^s(\xi)\diff \xi \diff u - P\right)\\
    &=  \mathbb{I}\left\{t \in \cTpar_{t+1}\right\}\cdot\left(  \int_{-C}^C \int_{-P}^{\xi + x_t^{\top}\theta^s}\diff u f^s(\xi)\diff \xi  - P\right)\\
    &=  \mathbb{I}\left\{t \in \cTpar_{t+1}\right\} \cdot \left(x_t^{\top}\theta^s + \int_{-C}^C \xi f^s(\xi) \diff \xi\right)\\
    &=   \mathbb{I}\left\{t \in \cTpar_{t+1}\right\} x_t^{\top}\theta^s
\end{align*}
where in the last equality we used that $\int_{-C}^C \xi f^s(\xi) \diff \xi = \mathbb{E}\left[\xi^s_t\right] = 0$. Thus, conditionally on $\cF_{t-1}$, $\tilde{y}_t -\tilde{x}_t^{\top}\theta^s$ is centered and it belongs to $[-P,P]$, which implies that it is $P$-subgaussian. Now, for all $t \in \{1,\dots, T\}$, we have
\begin{align*}
\widehat{\theta}^s_{t} &= 2P\left(\sum_{l\in \cTpar_{T+1}}x_l x_l^{\top}+  \bI_d\right)^{-1}\sum_{s\in \cTpar_{T+1}}\left(\mathbb{I}\left\{p_t \geq s_t\right\}-\frac{1}{2}\right)x_l\\
&= \left(\sum_{l < t}\tilde{x}_l \tilde{x}_l^{\top}+  \bI_d\right)^{-1}\sum_{l<t}\tilde{y}_l\tilde{x}_l.
\end{align*}

Using the fact that, for all $t \leq 1$, $\Vert \tilde{x}_t \Vert \leq B$ and $\Vert \theta^s \Vert \leq A$, and applying Theorem 2 by \citet{NIPS2011_e1d5be1c}, we find that for all $t \geq 0$, with probability $1-\delta$, 
$$\Vert\widehat{\theta}^s_t - \theta^s\Vert_{\sum_{s < t}\tilde{x}_l \tilde{x}_l^{\top}+  \bI_d}  \leq P\sqrt{d\log\left(\frac{1 + B^2T}{\delta}\right)} + A.$$
Note that our definitions of $\tilde{x}_t$ and $\tilde{y}_t$ ensure that 
\[
\Vert\widehat{\theta}^s_t - \theta^s\Vert_{\sum_{s\in \cTpar_{T+1}}x_l x_l^{\top} +  \bI_d} = \Vert\widehat{\theta}^s_t - \theta^s\Vert_{\sum_{s < t}\tilde{x}_l \tilde{x}_l^{\top}+  \bI_d}.\]
Moreover, for all $t$, 
$$\vert x_t^{\top}(\widehat{\theta}^s_t- \theta^s) \vert\leq \Vert x_t^{\top}\Vert_{\left(\sum_{s\in \cTpar_{T+1}}x_l x_l^{\top} +  \bI_d\right)^{-1}}\,\,\Vert\widehat{\theta}^s_t - \theta^s\Vert_{\mleft(\sum_{s\in \cTpar_{T+1}}x_l x_l^{\top} +  \bI_d\mright)}.$$
In particular, if $t \notin \cTpar_{T+1}$, $\Vert x_t^{\top}\Vert_{\left(\sum_{s\in \cTpar_{T+1}}x_l x_l^{\top} +  \bI_d\right)^{-1}}\leq \mu$, so with probability $1-\delta$,
$$\vert x_t^{\top}(\widehat{\theta}^s_t- \theta^s) \vert\leq \mu\left(P\sqrt{d\log\left(\frac{1 + B^2T}{\delta}\right)} + A\right).$$
For the choice  $\mu = \epsilon\left(P\sqrt{d\log\left(\frac{1 + B^2T}{\delta}\right)} + A\right)^{-1}$, this implies that
$$\vert x_t^{\top}(\widehat{\theta}^s_t- \theta^s) \vert\leq \epsilon,$$
which concludes the proof.
\end{proof}
\subsection{Proof of Lemma \ref{lem:eventF}}\label{app:proof lemma eventF}

\eventF*

\begin{proof} We control the error $\left\vert\widehat{I}^k - I\left(k\epsilon\right) \right\vert$ uniformly for $k \leq K$; the result for $\left\vert\widehat{J}^k - J(k\epsilon) \right\vert$ can be proved analogously.

For $k\leq K$, and $t\leq T$, let us define $\iota_t = \mathbb{I}\left\{t \in \cTint\right\}$, $y_t = 2P\mathbb{I}\left\{k\epsilon + x_t^{\top}\widehat{\theta}_t^b \leq p_l \leq b_t  \right\}$ and note that for $\cF_t = \sigma\left((x_1,\ldots, x_{t+1}, (\mathbb{I}\{s_1 \leq p_1\},\mathbb{I}\{b_1 \leq p_1\}),\ldots, (\mathbb{I}\{s_t \leq p_t\},\mathbb{I}\{b_t \leq p_t\}) \right)$, $\iota_t$ is $\cF_{t-1}$-measurable, and $y_t$ is $\cF_t$-adapted. Moreover,
\begin{align*}
\mathbb{E}\left[y_t\vert \cF_{t-1}\right] &= 2P\,\mathbb{P}\left[k\epsilon + x_t^{\top}\widehat{\theta}_t^b \leq p_t \leq b_t\, \vert\, \cF_{t-1}\right]\\
&= 2P\int_{k\epsilon + x_t^{\top}\widehat{\theta}^b_t}^{P}\int_{-C}^C\mathbb{I}\left\{x_t^{\top}\theta^b + \xi \geq u\right\}f^b(\xi)\diff \xi\frac{\diff u}{2P}\\
&= \int_{k\epsilon + x_t^{\top}\widehat{\theta}^b_t}^P D^b\left( u- x_t^{\top}\theta^{b}\right)\diff u.
\end{align*}
Using the change in variables $u' = u- x_t^{\top}\theta^{b}$, this implies that 
\begin{align*}
\mathbb{E}\left[y_t\vert \cF_{t-1} \right]&= \int_{k\epsilon + x_t^{\top}\widehat{\theta}^b_t - x_t^{\top}\theta^{b}}^{P- x_t^{\top}\theta^{b}} D^b\left(u'\right)\diff u'.
\end{align*}
Moreover, under Assumptions \ref{ass:bounded_context} and \ref{ass:bounded_densities}, $P-x_t^{\top}\theta^{b} \geq C$ and, for $u' \in [C,P-x_t^{\top}\theta^{b}]$, $D^b(u') = 0$. Thus,
\begin{align*}
\mathbb{E}\left[y_t\vert \cF_{t-1} \right]&= \int_{k\epsilon + x_t^{\top}\widehat{\theta}^b_t - x_t^{\top}\theta^{b}}^{C} D^b\left(u'\right)\diff u'\\
&=I(k\epsilon + x_t^{\top}\widehat{\theta}^b_t - x_t^{\top}\theta^{b}).
\end{align*}
Finally, note that $y_t - \mathbb{E}\left[y_t\vert \cF_{t-1} \right]$ is in $[-2P,2P]$. Then, using Lemma \ref{lem:Hoeffding_rigorous}, we find that
\begin{align*}
    \mathbb{P}\left(\vert \cTint \vert = \Tint \text{ and }\left\vert \widehat{I}^k - \frac{\sum_{t\in \cTint}I\left(k\epsilon + x_t^{\top}\left(\widehat{\theta}^b_t - \theta^b\right)\right)}{\Tint}\right\vert\geq 4P\sqrt{\frac{\log(1/\delta)}{2\Tint}}\right) \leq 2\delta.
\end{align*}
To conclude our proof, note that on the event $\cE_1$, for all $t \in \cTint$, $\vert x_t^{\top}(\widehat{\theta}^b_t - \theta^b)\vert \leq \epsilon$, so with probability $1-2\delta$, either $\vert \cTint \vert < \Tint$ or
\begin{align*}
\left\vert I\left(k\epsilon + x_t^{\top}\left(\widehat{\theta}^b_t - \theta^b\right)\right) - I(k\epsilon) \right \vert &=  \left\vert\int_{k\epsilon}^{k\epsilon + x_t^{\top}\left(\widehat{\theta}^b_t - \theta^b\right)} \left(1-F^b(\lambda)\right)d\lambda\right \vert\\
&\leq \left\vert\int_{k\epsilon}^{k\epsilon + x_t^{\top}\left(\widehat{\theta}^b_t - \theta^b\right)} d\lambda\right \vert\\
&\leq \mleft\vert x_t\left(\widehat{\theta}^b_t - \theta^b\right)\mright\vert\\
&\leq \epsilon.
\end{align*}
Using the same reasoning to control the error in estimating $J$, taking a union bound for $k \in \cK$, and using Lemma \ref{lem:eventE}, we find that on an event $\cE_2 \subset \cE_1$ of probability larger than $1-2\delta - 4\delta(2K+1)$, either $\vert \cTint \vert < \Tint$ or
\begin{align*}
    \left\vert \widehat{I}^k - I(k\epsilon) \right\vert\leq 4P\sqrt{\frac{\log(1/\delta)}{2\Tint}} + \epsilon
\end{align*}
and 
\begin{align*}
    \left\vert \widehat{J}^k - J(k\epsilon) \right\vert\leq 4P\sqrt{\frac{\log(1/\delta)}{2\Tint}} + \epsilon
\end{align*}
simultaneously for all $k\in \cK$. For the choice $\Tint = 8P^2\log(1/\delta)/\epsilon^2$, we obtain the desired result.
\end{proof}

\subsection{Proof of Lemma \ref{lem:Hoeffding_rigorous}}\label{app:proof hoeffding rigorous}

\Hoeffding*

\begin{proof}
Let us define $Z_t \defeq \sum_{l\leq t}\iota_l (y_l-\mathbb{E}\left[y_l \vert \cF_{l-1}\right])$, and for all $x\in \mathbb{R}$ let $M_t \defeq \exp\left\{x\,Z_t - \frac{1}{8}x^2(M-m)^2N_t\right\}$. We begin by showing that $M_t$ is a super-martingale. Indeed, we have that
\begin{align*}
    \mathbb{E}\left[\exp\{x\,\iota_t (y_t-\mathbb{E}\left[y_t \vert \cF_{t-1}\right])\} \,\,\Big \vert \,\,\cF_{t-1}\right] & = \mathbb{E}\left[\iota_t \exp\{x(y_t-\mathbb{E}\left[y_t \vert \cF_{t-1}\right])\}+ (1-\iota_t) \,\,\Big \vert \,\,\cF_{t-1}\right] \\
    &\leq \iota_t \exp\mleft\{\frac{x^2(M-m)^2}{8}\mright\}+ (1-\iota_t)\\
    &\leq \exp\mleft\{\frac{x^2(M-m)^2\iota_t}{8}\mright\},
\end{align*}
where we use the fact that $(y_t-\mathbb{E}\left[y_t\, \vert\, \cF_{t-1}\right])$ is bounded in $[m,M]$ together with the conditional version of Hoeffding's Lemma. Noticing that 
\begin{align*}
    M_t = M_{t-1}\exp\mleft\{x\,\iota_t (y_t-\mathbb{E}\left[y_t \vert \cF_{t-1}\right]) - \frac{x^2(M-m)^2\iota_t}{8}\mright\},
\end{align*}
this proves that $M_t$ is a super-martingale, and so $\mathbb{E}\left[M_t\right] \leq \mathbb{E}\left[M_0\right] = 1$.

\bigskip

\noindent Now, for all $\epsilon>0$ and all $l \in \mathbb{N}$, and all $x>0$, by a Markov-Chernoff argument,
\begin{align*}
    \mathbb{P}\left(Z_t \geq \epsilon \textnormal{ and }N_t = l\right) &= \mathbb{P}\left(\mathbb{I}\left\{N_t = l\right\}e^{xZ_t} \geq e^{\epsilon x}\right)\\
    &\leq e^{-\epsilon x}\mathbb{E}\left( \mathbb{I}\left\{N_t = l\right\}\cdot e^{xZ_t}\right)\\
    &=e^{-\epsilon x + \frac{x^2(M-m)^2l}{8}}\mathbb{E}\left(\mathbb{I}\left\{N_t = l\right\}\cdot e^{xZ_t - \frac{x^2(M-m)^2l}{8}}  \right).
\end{align*}
Using the previous result, we have that
\begin{align*}
    \mathbb{E}\left(  \mathbb{I}\left\{N_t = l\right\}\cdot e^{xZ_t - \frac{x^2(M-m)^2l}{8}}\right) &= \mathbb{E}\left( \mathbb{I}\left\{N_t = l\right\} \cdot e^{xZ_t - \frac{x^2(M-m)^2N_t}{8}} \right) \\
    &\leq \mathbb{E}\left(e^{xZ_t - \frac{x^2(M-m)^2N_t}{8}}  \right) \\
    &=\mathbb{E}(M_t)\\
    &\leq \mathbb{E}(M_0) = 1.
\end{align*}
This yields 
\begin{align*}
    \mathbb{P}\left(Z_t \geq \epsilon \textnormal{ and }N_t = l\right) &\leq e^{-\epsilon x + \frac{x^2(M-m)^2l}{8}}.
\end{align*}
In particular, for $\epsilon = (M-m)\sqrt{\frac{l\cdot\log(1/\delta)}{2}}$ and $x = \frac{4\epsilon}{l(M-m)^2}$,
\begin{align*}
    \mathbb{P}\left(Z_t \geq (M-m)\sqrt{\frac{l\cdot\log(1/\delta)}{2}} \textnormal{ and }N_t = l\right) &\leq \delta.
\end{align*}
This proves the first part of the Lemma. Summing over the values of $l$ from $1$ to $t$, we find that
\begin{align*}
    \mathbb{P}\left(Z_t \geq (M-m)\sqrt{\frac{N_t\log(1/\delta)}{2}} \text{ and }N_t\geq 1\right) &\leq t\delta.
\end{align*}
Similar arguments can be used to prove that 
\begin{align*}
    \mathbb{P}\left(-Z_t \geq (M-m)\sqrt{\frac{N_t\log(1/\delta)}{2}} \text{ and }N_t\geq 1\right) &\leq t\delta.
\end{align*}
In order to conclude the proof we observe that $Z_t = \hat{\mu}_tN_t$, we normalize by $N_t$, and observe that adding the case $N_t=0$ can only increase the probability.
\end{proof}
\subsection{Proof of Lemma \ref{lem:decompo_error}}\label{app: proof decompo error}
\decompoerror*

\begin{proof}
By Lemma \ref{lem:EGFT}, for any price $p = x_t^{\top}\widehat{\theta}^s_t + k\epsilon$ and $k'$ such that $(k, k')\in \cA_t$, and $\delta^s = p - x_t^{\top}\theta^s$, $\delta^b = p - x_t^{\top}\theta^b$, we have that
\begin{eqnarray*}
    \egft(x_t, p) = I(\delta^b) F^s(\delta^s) + J(\delta^s)D^b(\delta^b).
\end{eqnarray*}
Then,
\begin{align*}
\egft(x_t, p) =&\left(I(k' \epsilon) + I(\delta^b)- I(k' \epsilon)\right)\left(F^s(k \epsilon) + F^s(\delta^s) - F^s(k \epsilon)\right) + \\
    &\left(J(k\epsilon)+ J(\delta^s) - J(k\epsilon)\right)\left(D^b(k' \epsilon) + D^b(\delta^b) - D^b(k' \epsilon)\right).
\end{align*}
Moreover, by letting $\Delta I = I(\delta^b)- I(k' \epsilon)$, $\Delta F = F^s(\delta^s) - F^s(k \epsilon)$, $\Delta J = J(\delta^s) - J(k\epsilon)$ and $\Delta D = D^b(\delta^b) - D^b(k' \epsilon)$, this yields
\begin{align*}
\egft(x_t, p) =& I(k' \epsilon) F^s(k \epsilon ) + J(k\epsilon)D^b(k' \epsilon)+ I(\delta^b)\Delta F +\\&  F^s(k \epsilon )\Delta I + J(\delta^s)\Delta D + D^b(k' \epsilon )\Delta J.
\end{align*}

Since $I$ and $J$ are bounded by $2P$, and $F$ and $D$ are bounded by $1$, this implies that 
\begin{align*}
\left \vert \egft(x_t, p)  - \left(I(k' \epsilon) F^s(k \epsilon ) + J(k\epsilon)D^b(k' \epsilon)\right)\right \vert \leq 2P\vert \Delta F\vert +  \vert\Delta I\vert + 2P\vert\Delta D\vert + \vert\Delta J\vert. 
\end{align*}
Now, let us introduce $e^s = \delta^s - k \epsilon$, and $e^b = \delta^b - k' \epsilon$. Then, we have
\begin{align*}
    e^s &=  \left(x_t^{\top}\widehat{\theta}^s_t + k\epsilon\right) -\left(k\epsilon + x_t^{\top}\theta^s \right)\\
    &= x_t^{\top}(\widehat{\theta}^s_t-\theta^s).
\end{align*}
On event $\cE^{\textnormal{ETC}}$ we have $\vert e^s\vert \leq \epsilon$. Since $F$ is $L$-Lipschitz continuous, and $\Delta F = F^s(k \epsilon + e^s) - F^s(k \epsilon)$, this implies $\vert \Delta F \vert \leq L\epsilon$. Similarly, $J$ is $1$-Lipschitz continuous, and $\Delta J = J(k \epsilon + e^s) - J(k\epsilon)$, so $\vert \Delta J \vert \leq \epsilon$. 
Similarly, 
\begin{align*}
    e^b &= \left(x_t^{\top}\widehat{\theta}^s_t + k\epsilon\right) -  \left(x_t^{\top}\theta^b + k'\epsilon\right)\\
    &=  \left(x_t^{\top}(\widehat{\theta}^s_t-\widehat{\theta}^b_t) + k\epsilon\right) - k'\epsilon + x_t^{\top}(\widehat{\theta}^b_t -\theta^b ).
\end{align*}
By definition of $\cA_t$, under the event $\cE^{\textnormal{ETC}}$ we have $\vert e^b \vert \leq 2\epsilon$. Since $D$ is $L$-Lipschitz continuous and $\Delta D = D^b(k' \epsilon + e^b) - D^b(k' \epsilon)$,  $\vert \Delta D \vert \leq 2L\epsilon$. Similarly, $I$ is $1$-Lipschitz continuous, $\Delta I = I(k' \epsilon + e^b)- I(k' \epsilon)$, so this implies that $\vert\Delta I\vert \leq 2\epsilon$.
Putting everything together, we find that on the event $\cE^{\textnormal{ETC}}$,
\begin{align*}
\left \vert \egft(x_t, p)  - \left(I(k' \epsilon) F^s(k \epsilon ) + J(k\epsilon)D^b(k' \epsilon)\right)\right \vert \leq& 6PL\epsilon +3\epsilon.
\end{align*}

Similarly, denoting $\widehat{\Delta}I = I(k' \epsilon) -\widehat{I}^{k'}$, $\widehat{\Delta}D = D(k' \epsilon) -\widehat{D}^{k'}$, $\widehat{\Delta}J = J(k \epsilon) -\widehat{J}^{k}$ and $\widehat{\Delta}F = F(k \epsilon) -\widehat{F}^{k}$, we have
\begin{align*}
\left \vert I(k' \epsilon) F^s(k \epsilon ) + J(k\epsilon)D^b(k' \epsilon) - \left(\widehat{I}^{k'}\widehat{F}^k + \widehat{J}^k\widehat{D}^{k'}\right)\right \vert &\leq 2P\vert \widehat{\Delta} F\vert +  \vert\widehat{\Delta} I\vert + 2P\vert\widehat{\Delta} D\vert + \vert\widehat{\Delta}J\vert.
\end{align*}
This concludes the proof.
\end{proof}
\section{Motivating example}\label{sec:example}

Lemma \ref{lem:EGFT} emphasizes that the expected gain from trade at a given price $p$ depends on the quantities $\delta^s = p - x_t^{\top}\theta^s$ and $\delta^b = p - x_t^{\top}\theta^b$. Remember that the difference in average valuations $\Delta$ is given by $\Delta = x_t^{\top}\theta^b - x_t^{\top}\theta^s$, and with this notation, $\delta^b = \delta^s - \Delta$. Therefore, the expected gain from trade can be rewritten as a function of the pair $(\delta^s, \Delta)$. As an immediate consequence, we see that the optimal increment $\delta^s$ only depends on the difference in average valuations $\Delta$: if $p = x^{\top}\theta^s + \delta^s$ maximizes $\egft(x,p)$, and if $x'$ is such that $x'^{\top}\theta^b - x'^{\top}\theta^s = x^{\top}\theta^b - x^{\top}\theta^s$, then $p' = x'^{\top}\theta^s + \delta^s$ maximizes $\egft(x',p')$.

On the other hand, the following example shows that there no explicit dependence of the optimal price increment $\delta^s$ on the difference in average valuations $\Delta$. In words, when $\Delta$ is small, we might prefer to choose an increment $\delta^s$ that leads to trade happening with lower probabilities but corresponds to higher rewards. By contrast, as $\Delta$ increases, similar values of the increment $\delta^s$ will correspond to higher gains if the trade happens. Then, we might choose to post prices corresponding to a smaller increment $\delta^s$ to increase the probability that the trade happens.

This reasoning demonstrates that knowing that an increment $\delta^s$ is optimal for a difference in average valuations $\Delta$ does not allow us to determine the optimal increment $\delta'$ corresponding to a different value $\Delta'$ of the difference in average valuations. This implies that to precisely identify the optimal price increment $\delta^s$ for all differences in average valuations $\Delta$, it may be necessary to have accurate estimates of the functions $F^s$ and $I$ for a broad range of values of $\delta^s$. Similar arguments can be employed to argue that precise estimates of the value of $D^b$ and $J$ for a wide range of values of $\delta^b = \Delta - \delta^s$ might also be necessary.

To illustrate this phenomenon, we construct an example where different levels of $\Delta$ lead to entirely different choices of the optimal price increment $\delta^s$. Specifically, we consider a scenario where, for certain values of $(s_1, s_2, s_3)\in \mathbb{R}^3$, $(b_1, b_2, b_3)\in \mathbb{R}^3$, $(\alpha_1, \alpha_2, \alpha_3)$ in the simplex, and $\theta > 0$ (to be defined later), the density $f^s$ (resp. $f^b$) of the seller's noise $\xi^s$ (resp. buyer's noise $\xi^b$) is given by:
\begin{align*}
f^s(\delta) &= \frac{\mathbb{I}\{x \in [s_1, s_1 + \theta] \cup [s_2, s_2 + \theta]\cup [s_3, s_3 + \theta]\}}{3\theta}\\
f^b(\delta) &= \frac{\alpha_1\mathbb{I}\{x \in [b_1, b_1 + \theta]\} +\alpha_2\mathbb{I}\{x \in [b_2, b_2 + \theta]\} + \alpha_3\mathbb{I}\{x \in [b_3, b_3 + \theta]\}}{\theta}.
\end{align*}
The setting in illustrated in figure \ref{fig:example}.

We assume that $(s_1, s_2, s_3)$ and $(b_1, b_2, b_3)$ verify $s_1 + \theta < b_1$, $b_1+ \theta < s_2$, $s_2 + \theta < b_2$, $b_2 + \theta < s_3$, and $s_3 + \theta < b3$. Then, it is straightforward to see that while $\Delta > 0$, and $\Delta$ is small enough so that $b_1+ \theta + \Delta < s_2$ and $b_2 + \theta + \Delta < s_3$, the optimal increment belongs to the set $\{\delta_1, \delta_2, \delta_3\}$, where $\delta_1$, $\delta_2$, and $\delta_3$ belong respectively to $[s_1 + \theta, b_1]$, $[s_2 + \theta, b_2]$, and $[s_3 + \theta, b_3]$. 

We also assume that $\alpha_1$ is much larger than $\alpha_2$, which in turn is much larger than $\alpha_3$. Then, the increment $\delta_1$ is such that the trade happens with the highest probability: indeed, if $\delta^s = \delta^b > b_1$, the buyer rejects the trade with a high probability. For the same reasons, the probability of a trade happening at an increment $\delta_2$ is much lower, and the probability of a trade happening at increment $\delta_3$ is the lowest.

Finally, we assume that $b_3 - s_3$ is much larger than $b_2 - s_2$, which in turn is much larger than $b_1 - s_1$. Then, the gain from any trade happening at increment $\delta_1$ is small compared to the gain from trades happening at increment $\delta_2$, which in turn is small compared to the gain from a trade happening at increment $\delta_3$.

\begin{figure}[!ht]
    \includegraphics[width=\textwidth]{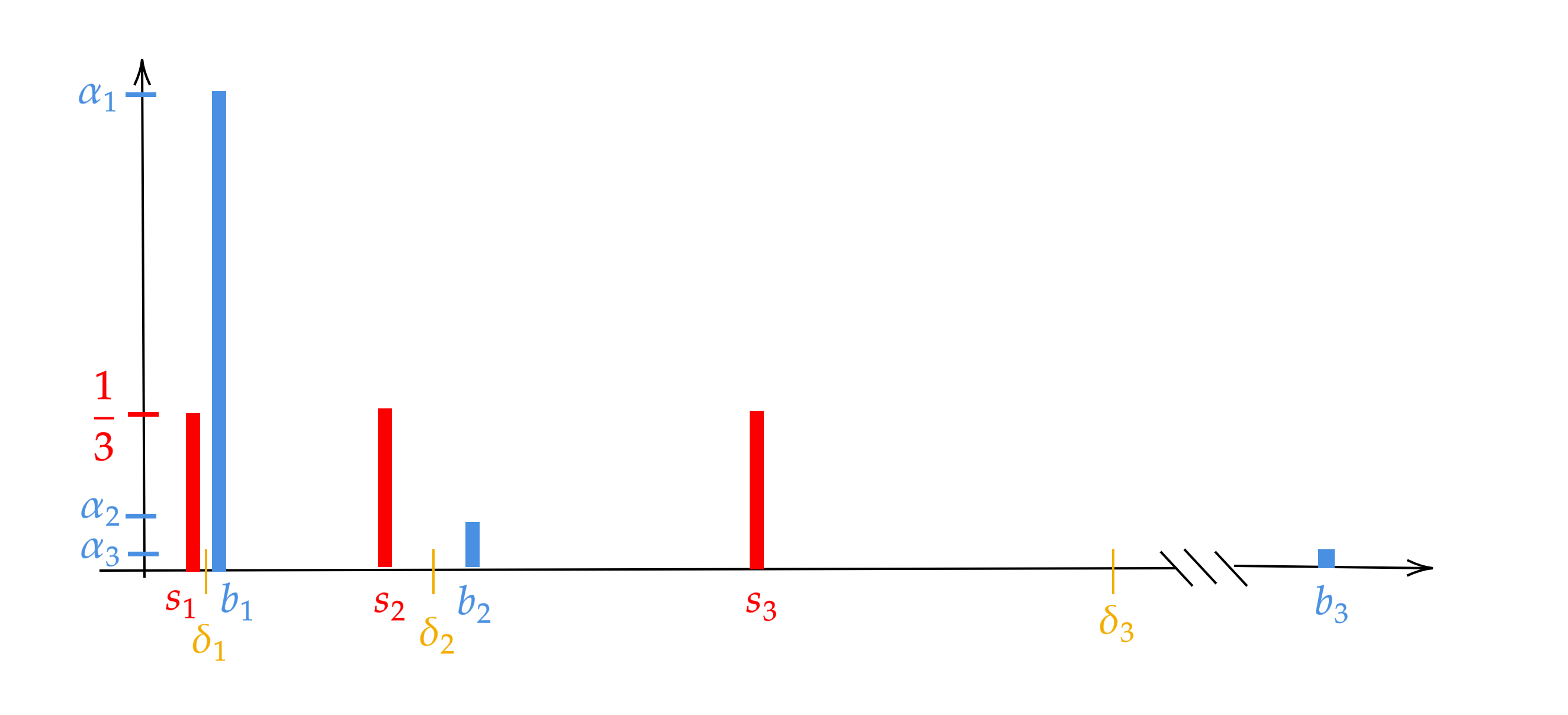}
    \caption{Illustration for the c.d.f. $f^s$ and $f^b$ corresponding to the example described in Section \ref{sec:example}. The c.d.f. $f^s$ is represented in red, the c.d.f. $f^b$ is represented in blue.\label{fig:example}}
\end{figure}

For some well-chosen values of the parameters specified below, when the average gain from trade of the seller and the buyer are both zero ($x_t^{\top}\theta^s = x_t^{\top}\theta^b = 0$), the most profitable increment is $\delta_3$: the probability of the trade happening is less likely, but when it occurs, it is more profitable.

Now, to obtain the probability of a price $p = x^{\top}\theta^s + \delta^s$ being accepted by the buyer, it is sufficient to translate the cumulative distribution function (c.d.f.) of $\xi^b$ (represented in Figure \ref{fig:example}) by the quantity $\Delta = x^{\top}(\theta^b - \theta^s)$. In particular, when $\Delta > 0$ is small enough so that $b_1 + \theta + \Delta < s_2$ and $b_2 + \theta + \Delta < s_3$, we observe that, on the one hand, the probability that the trade happens for the increments $\delta^s \in {\delta_1, \delta_2, \delta_3}$ remains unchanged; however, the corresponding gains if the trade occurs all increase with $\Delta$. For some well-chosen values of the parameters, if $\Delta$ is positive but sufficiently small, the expected gain from trade is maximized by $\delta^s = \delta^2$. When $\Delta$ is large, the expected gain from trade is maximized by choosing $\delta^s = \delta^1$: in other words, this choice of increment ensures that the sale happens with the highest probability, and each sale leads to a reward of at least $\Delta$.
In Figure \ref{fig:EFGT}, we plot the expected gain from trade for different values of $\Delta$ as a function of $\delta^s$. The values chosen for the parameters are as follows: $(s_1, s_2, s_3) = (0, 2, 6)$, $(b_1, b_2, b_3) = (0.01, 3, 20)$, $(\alpha_1, \alpha_2, \alpha_3) = (0.85, 0.11, 0.04)$, and $\theta = 0.001$. With these values, for $\Delta = 0$, the optimal increment is $\delta^3 = 10$; for $\Delta = 1$, it is $\delta^2 = 2.5$; and for $\Delta = 1.5$, it is $\delta^3 = 0.01$.
\begin{figure}[!ht]
    \includegraphics[width=\textwidth]{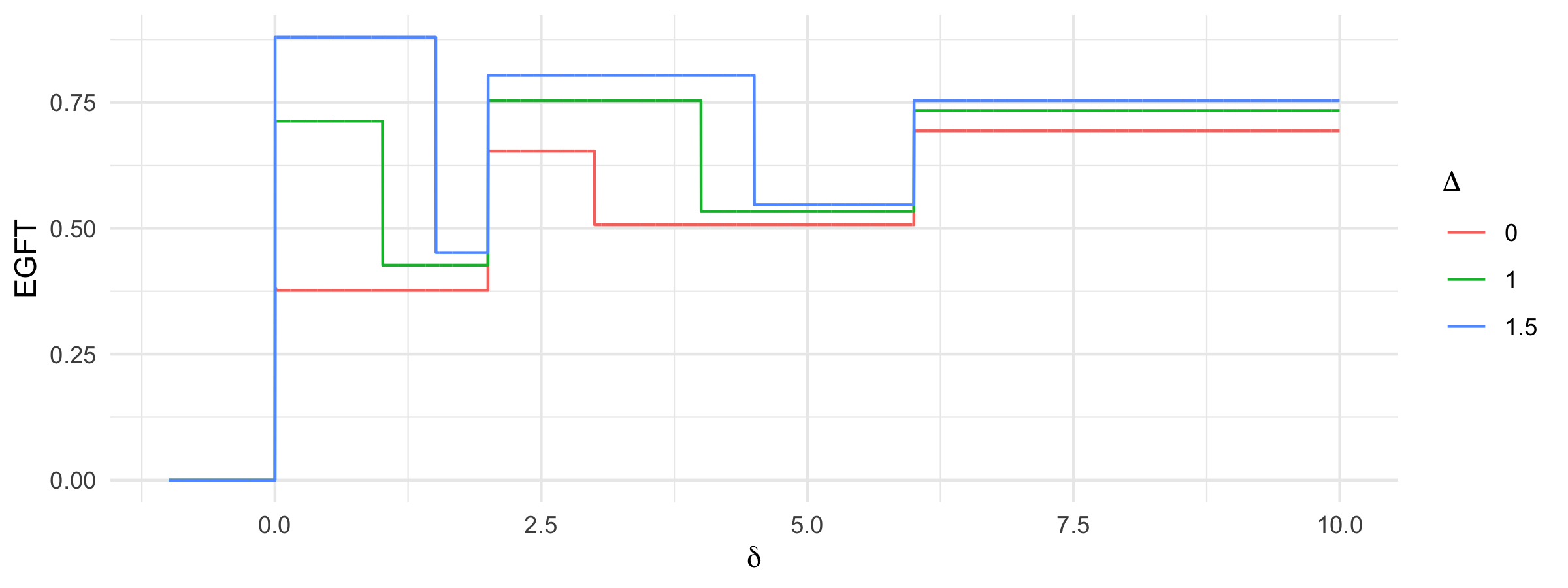}
    \caption{Expected gain from trade for different values of increment $\delta^s$ and difference in average valuations $\Delta$.\label{fig:EFGT}}
\end{figure}

\section{Proofs of \Cref{sec:one bit}}

\profitLemma*

\begin{proof}
To simplify notation, and since we focus on a single round $t\in\cTB$, we omit the explicit dependence on $t$ from $p_t$, $q_t$, $i_t$, $s_t$, and $b_t$.
We consider two cases. First, if $b-s\le \nicefrac2T$, the inequality is immediately satisfied since the lhs is non-negative while the rhs is $\le 0$.
Second, if $b-s>\nicefrac2T$, let $\cE$ be the event in which $p\in[s,(s+b)/2]$. Since $p\sim \cU([0,P])$ we have $\pr(\cE)=(b-s)/2P$. Moreover, we have 
\[
\pr\mleft(i=\mleft\lfloor \log \frac{1}{b-p}\mright\rfloor\mid A\mright)=\frac{1}{\log T},
\]
which is well defined since
\[
p\le \frac{s+b}{2}\le b-\frac{1}{T}.
\]
Let $\cE'$ denote the event in which $i=\lfloor \log (1/(b-p))\rfloor$. Under $\cE$ and $\cE'$ we get
\[
q=p+2^{-i}=p+2^{-\lfloor \log (1/(b-p))\rfloor}\ge p+2^{-(\log (1/(b-p)) +1)}=p+\frac{b-p}{2}=\frac{p+b}{2}.
\]
Therefore, when $b\ge s+\nicefrac{2}{T}$,
\begin{align*}
\E{\prof(p,q)}&=(q-p)\,\frac{b-s}{2P\log T}\\
&\ge \mleft( \frac{p+b}{2}-p \mright)\,\frac{b-s}{2P\log T}\\
& \ge \frac{(b-s)^2}{8P\log T}.
\end{align*}
This concludes the proof.
\end{proof}

\logTgft*
\begin{proof}
    First, we observe that given the sequence $(b_t,s_t)_{t \in \range{T}}$, and a time step $\tau$, by Azuma-Hoeffding inequality we have that, with probability at least $1-\delta$,
    \[  \sum_{t \in \range{\tau}} \prof_t(p_t,q_t) \ge \E{\prof_t(p,q)} - \sqrt{2 \log(1/\delta) \sum_{t \in \range{\tau}} ([b_t-s_t]^+)^2 }. \]
    Then, by applying a union bound, the above inequality holds simultaneously for all possible $\tau$ with probability at least $1-\delta T$. Then, by setting $\delta=T^2$, with probability at least $1-1/T$, it holds: 
\begin{align*}
     & \sum_{t \in \range{\tau}}  \prof_t(p_t,q_t) \\
     & \ge \sum_{t \in \range{\tau}}  \E{\prof_t(p,q)} - \sqrt{4 \log(T) \sum_{t \in \range{\tau}} ([b_t-s_t]^+)^2 }\\
     & \ge \sum_{t \in \range{\tau}} \mleft( \frac{([b_t-s_t]^+)^2}{8P\log T}-\frac{2}{T} \mright) - \sqrt{4 \log(1/\delta) \sum_{t \in \range{\tau}} ([b_t-s_t]^+)^2 } & \textnormal{(by \Cref{lemma:profit})}\\
     & \ge \sum_{t \in \range{\tau}}  \frac{([b_t-s_t]^+)^2}{8P\log T} - \sqrt{4 \log(T) \sum_{t \in \range{\tau}} ([b_t-s_t]^+)^2 } -2\\
     & =\tau \sum_{t \in \range{\tau}}  \frac{1}{\tau}\frac{([b_t-s_t]^+)^2}{8P\log T} - \sqrt{4 \log(T) \sum_{t \in \range{\tau}} ([b_t-s_t]^+)^2 } -2\\
     & \ge  \frac{\tau}{8P\log T} \mleft(\sum_{t \in \range{\tau}}  \frac{[b_t-s_t]^+}{\tau}\mright)^2 - \sqrt{4 \log(T) \sum_{t \in \range{\tau}} ([b_t-s_t]^+)^2 } - 2& \textnormal{(by Jensen's Inequality)}\\
     %
     &\ge \frac{\alpha}{8 P \log(T)} \sum_{t\in \range{\tau}} [b_t-s_t]^+ - \sqrt{4 \log(T) \sum_{t \in \range{\tau}} ([b_t-s_t]^+)^2 } -2& \textnormal{(by \Cref{assumption:active market})}\\
     &\ge \frac{\alpha}{8 P \log(T)}  \sum_{t\in \range{\tau}} [b_t-s_t]^+ - \sqrt{4 P^2 \log(T)\sum_{t \in \range{\tau}} [b_t-s_t]^+ } -2.
\end{align*}
This concludes the proof.
\end{proof}

\oneBitThm*
\begin{proof}
The one-bit algorithm is global budget balanced by construction (see choice of $\budget$).

Then, we condition the high probability regret bound on the following events:
\begin{itemize}
    \item With probability $1-1/T$ the two-bit EOC algorithm guarantees a number of exploration rounds smaller than $|\cte|$ and regret at most $R^{(2)}_T$;
    \item With probability $1-1/T$, it holds the inequality in \Cref{lemma:logT gft}
    \item By Azuma-Hoeffding, with probability at least $1-1/T$ it holds
    \[ \sum_{t=1}^\tau \max_p \egft(x_t,p)\le \sum_{t=1}^\tau [b_t-s_t]^+ + \sqrt{16P^2\tau\log(T)}.\]
\end{itemize}

Then, the regret can be bounded as follows 
\begin{align*}
    R_T&=\sum_{t=1}^T\max_p \egft(x_t,p)-\sum_{t=1}^T\egft(x_t,p_t)\\
    &\le \sum_{t=1}^\tau \max_p \egft(x_t,p) + \sum_{t=\tau+1}^T \mleft(\max_p\egft(x_t,p)-\egft(x_t,p_t)\mright)\\
    &\le \sum_{t=1}^\tau \max_p \egft(x_t,p) + 2|\cte| + R^{(2)}_T\\
    & \le \sum_{t=1}^\tau [b_t-s_t]^+ 
    + \sqrt{16P^2\tau\log(T)} + 2|\cte| + R^{(2)}_T\\
    & \le \sum_{t=1}^\tau [b_t-s_t]^+ 
    + \alpha \tau + 2|\cte| + R^{(2)}_T\\
    & \le 2\sum_{t=1}^\tau [b_t-s_t]^+  + 2|\cte| + R^{(2)}_T,
\end{align*}
where in the second-to-last inequality we use $\tau\ge 1024P^3\alpha^{-2}\log^3T$ (since $\tau\ge \budget/2P$), and the last inequality is by \Cref{assumption:active market}.

Then, by \Cref{lemma:logT gft} and since $\sum_{t=1}^\tau [b_t-s_t]^+\ge \sum_{t=1}^\tau \prof_t(p_t,q_t)\ge 2048P^4\alpha^{-2}\log^3T $, we have
\begin{align*}
\sum_{t \in \range{\tau}} \prof_t(p_t,q_t) &\ge
    \frac{\alpha}{8 P \log(T)}  \sum_{t\in \range{\tau}} [b_t-s_t]^+ - \sqrt{4 P^2 \log(T)\sum_{t \in \range{\tau}} [b_t-s_t]^+ } -2\\
    &\ge \frac{\alpha}{8 P \log(T)}  \sum_{t\in \range{\tau}} [b_t-s_t]^+ - \sqrt{8 P^2 \log(T)\sum_{t \in \range{\tau}} [b_t-s_t]^+ } \\
    &\ge \frac{\alpha}{16 P \log(T)}  \sum_{t\in \range{\tau}} [b_t-s_t]^+.
\end{align*}

Then, since  $\budget\le 2048P^4\alpha^{-2}|\cte|\log^3T$ and $\sum_{t=1}^\tau\prof_t(p_t,q_t)\le \budget+1$, by substituting in the expression above we obtain the desired bound on $\sum_{t\in \range{\tau}} [b_t-s_t]^+$ and hence on the regret. 
\end{proof}

\end{document}